\documentclass[a4paper,11pt]{article}
\usepackage[left=2.5cm,right=2.5cm,top=3cm,bottom=3cm]{geometry}

\usepackage[persection]{belov_paper}
\numberwithin{equation}{section}

\usepackage{euscript}

\usepackage{mathtools}
\usepackage{ytableau}

\usepackage{xcolor}

\newcommand{\beE}{\begin{equation}}
\newcommand{\enE}{\end{equation}}
\def\beA#1\enA{\begin{align}#1\end{align}}
\def\beM#1\enM{\begin{multline}#1\end{multline}}

\usepackage[framemethod=TikZ]{mdframed}
\usepackage{chngcntr}

\newmdenv[%
  roundcorner=5pt,
  linecolor=blue!15,
  linewidth=2pt,
  subtitlebackgroundcolor=blue!15,
  subtitleaboveskip=0pt,
  subtitlebelowskip=0pt,
  subtitleinneraboveskip=0pt,
  innerbottommargin=0pt,
  subtitlefont=\normalfont
]{mdfigure}

\def\myfigureInternal#1#2#3{
\refstepcounter{figure} #1
\begin{mdfigure}[
  frametitle={
    \tikz[baseline=(current bounding box.east),outer sep=0pt]
    \node[anchor=east,rectangle,fill=blue!15,rounded corners]
    {Figure~\thefigure};},
  frametitleaboveskip=-10pt
  innertopmargin=0pt,
  innerbottommargin=0pt,
  roundcorner=5pt,
  linecolor=blue!15,
  linewidth=2pt,
  subtitlebackgroundcolor=blue!15,
  subtitleaboveskip=0pt,
  subtitlebelowskip=0pt,
  subtitleinneraboveskip=0pt,
  subtitlefont=\normalfont
]
#3
\mdfsubtitle{\medskip #2}
\end{mdfigure}
}

\def\myfigure#1#2#3{
\begin{figure}[tb]
\myfigureInternal{#1}{#2}{#3}
\negbigskip
\end{figure}
}

\usepackage{tikz}
\usepackage{pgfplots}
\usetikzlibrary{graphs, shapes.geometric,matrix,arrows.meta, positioning}
\tikzset{>={Stealth[width=2mm,length=2mm]}}

\newcommand{\sqtimes}{\mathbin{\ooalign{$\sqcup$\cr$\hidewidth\times\hidewidth$}}}

\title{Tight Quantum Lower Bound for Approximate Counting with Quantum States}
\author{
Aleksandrs Belovs\thanks{Faculty of Computing, University of Latvia}
\and
 Ansis Rosmanis\thanks{Graduate School of Mathematics, Nagoya University, Japan}
}
\date{\negbigskip\negbigskip}

\newcommand{\inn}{\bullet}
\newcommand{\out}{\circ}
\newcommand{\Spe}[1]{\mathcal{S}^{#1}}

\begin{document}

\maketitle

\begin{abstract}
We prove tight lower bounds for the following variant of the counting problem considered by Aaronson \etal~\cite{aaronson:counting}.
The task is to distinguish whether an input set $x\subseteq [n]$ has size either $k$ or $k'=(1+\eps)k$.
We assume the algorithm has access to
\begin{itemize}
\item the membership oracle, which, for each $i\in [n]$, can answer whether $i\in x$, or not; and
\item the uniform superposition $\ket|\psi_x> = \sum_{i\in x} \ket|i>/\sqrt{|x|}$ over the elements of $x$. 
Moreover, we consider three different ways how the algorithm can access this state:
    \begin{itemize}
    \item the algorithm can have copies of the state $\ket|\psi_x>$;
    \item the algorithm can execute the reflecting oracle which reflects about the state $\ket|\psi_x>$;
    \item the algorithm can execute the state-generating oracle (or its inverse) which performs the transformation $\ket|0>\mapsto\ket|\psi_x>$.
    \end{itemize}
\end{itemize}
Without the second type of resources (the ones related to $\ket|\psi_x>$), the problem is well-understood, see Brassard \etal~\cite{brassard:counting}.
The study of the problem with the second type of resources was recently initiated by Aaronson \etal~\cite{aaronson:counting}.

We completely resolve the problem for all values of $1/k \le \eps\le 1$, giving tight trade-offs between all types of resources available to the algorithm.
We also demonstrate that our lower bounds are tight.
Thus, we close the main open problems from~\cite{aaronson:counting}.

The lower bounds are proven using variants of the adversary bound from~\cite{belovs:variations} and employing representation theory of the symmetric group applied to the $S_n$-modules $\bC^{\binom{[n]}k}$ and $\bC^{\binom{[n]}k}\otimes \bC$.
\end{abstract}

\tableofcontents

\clearpage

\section{Introduction}

\subsection{Motivation}

The theory of quantum query algorithms deals extensively with the standard input oracle, which is the canonical quantum counterpart of the usual deterministic input oracle.
Other kinds of input oracles have received much less attention, especially from the lower bound point of view.
We find this unfortunate, since these questions are not only interesting theoretically, but also come up in practice.

There are two main techniques for proving quantum query lower bounds: the polynomial method and the adversary method.
The polynomial method was developed by Beals, Buhrman, Cleve, Mosca, and de Wolf~\cite{beals:pol}, and it has been applied to a large number of problems afterwards.
The adversary bound was first formulated by Ambainis~\cite{ambainis:adv} and slightly generalised by H{\o}yer, Neerbek, and Shi~\cite{hoyer:orderedSearch} to what we call a positive-weighted version of the bound, which has been used extensively ever after.
Afterwards, a significant strengthening was obtained by H\o yer, Lee, and \v Spalek~\cite{hoyer:advNegative}: the so-called negative-weighted version of the bound%
\footnote{%
Note that ``negative-weighted'' here means that negative entries in the adversary matrix are allowed, but not imposed.  
This formulation is also commonly known as the general version, but we keep the term ``general'' for the version of~\cite{belovs:variations} that allows general unitary input oracles.
}.
The latter version was shown to be tight by Reichardt~\cite{reichardt:spanPrograms, reichardt:advTight}.
It was further generalised in works by Ambainis, Magnin, R\"otteler, and Roland~\cite{ambainis:symmetryAssisted}; and Lee, Mittal, Reichardt, \v Spalek, and Szegedy~\cite{lee:stateConversion} to include further problems besides function evaluation.

However, all these versions of the adversary bound assume the standard input oracle.
Motivated by the task of distinguishing the main ideas at the heart of the quantum adversary method from the details specific to the standard input oracle, Belovs~\cite{belovs:variations} developed version of the bound that allows arbitrary unitaries as input oracles.  
The preceding versions of the bound could then be obtained as special cases thereof.
One thing missing from~\cite{belovs:variations} though are applications of the developed techniques to actual problems.
This brings us to our main motivation behind this paper: can the general version of the adversary bound be applied to real-world problems?
We show that this is indeed possible on the case of the counting problem.

\subsection{The Approximate Counting Problem}
\label{sec:counting}
The main problem under consideration in this paper is approximate counting.
Given some subset\footnote{%
It is customary to denote input strings to quantum algorithms using lower case Latin letters like $x$ or $y$, with $x_i$ denoting individual symbols of the input string.
We continue with this tradition, and, while we mostly think of the input $x$ as a subset of $[n]$ rather then the corresponding bit-string, we still denote it by a lower case Latin letter.
}
$x\subseteq[n]$ estimate its size with multiplicative precision $\eps$.
The decision version of the problem --- distinguishing the cases when $x$ has size exactly $k$ or exactly $(1+\eps)k$ --- is more appropriate for lower bounds.
Clearly, a lower bound for the latter also gives a lower bound for the former.

The canonical way of encoding $x$ is via its characteristic bit-sting $(x_i)_{i\in[n]}$ with $x_i=1$ iff $i\in x$.
The corresponding input oracle is called the membership oracle because a query tells whether a specific element is a member of the input set.

Quantum query complexity of approximate counting with access to the membership oracle was settled down early on in the history of quantum computation.
Brassard, H{\o}yer, and Tapp showed 
that $O\s[\frac1\eps \sqrt{n/|x|}]$ queries to the membership oracle suffice for approximate counting~\cite{brassard:counting}.
The matching lower bound of $\Omega\s[\frac1\eps \sqrt{n/|x|}]$ was obtained by Nayak and Wu~\cite{nayak:approximatingMedian} using the polynomial method.
A simple proof can be obtained using the positive-weighted adversary.
The original paper by Ambainis~\cite{ambainis:adv} already contains the proof for the special case of $k=n/2$.  The general case has essentially the same proof, which can be found in \rf{app:posWeighted}.

But the membership oracle is not the only way how to encode the input set $x$.
Aaronson, Kothari, Kretschmer, and Thaler~\cite{aaronson:counting} raised the question of estimating complexity of approximate counting when the quantum algorithm has access not only to the membership oracle, but also to the uniform superposition of the elements of the set:
\begin{equation}
\label{eqn:psix}
\psi_x = \frac1{\sqrt{|x|}}\sum_{i\in x}\ket|i>.
\end{equation}
This requires clarification: what does it mean to have access to $\psi_x$?
Ref.~\cite{aaronson:counting} assumes the following two models:
the algorithm can have copies of the state $\psi_x$ or it can reflect about $\psi_x$.
Possession of copies is a rather standard assumption.
The reflecting oracle is more unusual, but it is helpful in algorithms, as it can be used for amplitude amplification and estimation, see \rf{app:upper} for more detail.

The main result of~\cite{aaronson:counting} is as follows.
In order to distinguish whether the size of the input set $x$ is $k$ or $2k$, the quantum algorithm either has to invoke the membership oracle $\Omega(\sqrt{n/k})$ times or access the state $\psi_x$ at least
$\Omega\sB[\min\{ k^{1/3}, \sqrt{n/k} \}]$ times in the aforementioned way.
It was shown to be optimal in the sense that $O(\sqrt{n/k})$ membership queries \emph{alone} suffice to solve this problem, as well as $O\sB[\min\{ k^{1/3}, \sqrt{n/k} \}]$ accesses to the state $\psi_x$ \emph{alone} suffice.
Thus, nothing can be gained by combining the two resources.

The following open problems were formulated in~\cite{aaronson:counting}.
The first one was to distinguish the cases $|x|=k$ and $|x|=(1+\eps)k$ for $\eps\ll 1$.
The second one was to determine the complexity of the problem when the algorithm only has access to copies of the state $\psi_x$, without having access to the reflecting oracle.

\subsection{Our Results}
In this paper, we completely resolve these problems for all values of $\eps$ between 0 and 1, and go beyond that.
In addition to accessing $\psi_x$ via copies and the reflecting oracle, we also allow the state-generating oracle, which performs the transformation $\ket|0>\mapsto \ket|\psi_x>$ for some predetermined state $\ket|0>$  (as it is customary, we also allow to run this transformation in reverse).
The state-generating oracle encompasses both of the aforementioned models of accessing $\psi_x$ from Ref.~\cite{aaronson:counting}.%
\footnote{
    In the second version of their paper~\cite{aaronson:counting}, Aaronson \etal also considered a state-generating oracle.  However, due to technical reasons, the oracle prepared two copies of the state at once: $\ket|0>\mapsto \ket|\psi_x>\ket|\psi_x>$.
    Such an oracle can be simulated with two executions of the standard state-generating oracle $\ket|0>\mapsto \ket|\psi_x>$, while simulation in the opposite direction is not clear.
}
Indeed, one invocation of the state-generating oracle suffices to get a copy of $\psi_x$, while two invocations (one direct and one reverse) suffice to reflect about $\psi_x$.
On the other hand, it is hard to simulate the state-generating oracle using just copies and reflections.

\newcommand{\QL}{\ell}
\newcommand{\QM}{q_{\mathsf{M}}}
\newcommand{\QR}{q_{\mathsf{R}}}
\newcommand{\QG}{q_{\mathsf{G}}}

The results of our paper are summarised in \rf{tbl:main}.
We show that the algorithm has one of the eight options to solve the problem.
In each of the options, the algorithm either uses one type of resources (copies of the state, or one of the oracles), or a pair of resources.
In the case of a single resource, we state the corresponding lower bound; and in the case of a pair of resources, we state a trade-off between them.
It is not possible to gain anything by combining more than two types of resources:
If the algorithm uses three or all four types of resources to solve the problem, then among them there exists a pair or a single resource that satisfies one of the eight conditions of \rf{tbl:main}.
The Rows 6 and 7 of the table do seem to use a triple of resources, but because the number of copies and the number of invocations of the state-generating oracle are joined by a sum: $\ell+\QG$, each of these rows can be further broken down into two trade-offs between pairs of resources.
From the algorithmic point of view, the state-generating oracle is only used to prepare copies of the state in this case.

\begin{table}[tbh]
\[
\begin{tabular}{|l|r@{$\;=\;$}l|}
\hline
Copies of the state & 
$\ell$ &
$\displaystyle \Omega\sC[\min\sfigB{k,\; \frac{\sqrt{k}}\eps,\; \frac{n}{k\eps^2} } ]$\\

\hline
Membership oracle & 
$\QM$ &
$\displaystyle \Omega \s[\frac1\eps\sqrt{\frac nk}]$ \\

\hline
State-generating oracle &
$\QG$ &
$\displaystyle\Omega\sD[\min\sfigC{\frac1\eps \sqrt{\frac nk},\;\; \;\; \frac{k^{1/3}}{\eps^{2/3}} }]$ \\
    \qquad $\bullet$ with copies of the state &
$\QG\sqrt{\ell}$ &
$\displaystyle \Omega\sC[\frac{\sqrt k}\eps] $\\

\hline
Reflecting oracle &
$\QR$ &
$\displaystyle \Omega \s[\min\sfigC{\frac1\eps\sqrt{\frac nk},\; \sqrt{\frac k\eps} + \sqrt{\frac nk}}]$ \\

\qquad \parbox{12em}{$\bullet$ with copies of the state\\ \strut\quad or state-generating oracle} &
$\QR\sqrt{\ell + \QG}$ &
$\displaystyle \Omega\sC[\frac{\sqrt k}\eps] $\\

\qquad \parbox{12em}{$\bullet$ with copies of the state\\ \strut\quad or state-generating oracle} &
$\QR$ &
$\displaystyle \Omega\s[\sqrt{\frac k\eps}]$  \quad and\quad $\displaystyle\ell+\QG\ge 1$\\

\qquad $\bullet$ with membership oracle: &
$\QR$ &
$\displaystyle \Omega\s[\sqrt{\frac k\eps}] $ \quad and\quad $\displaystyle \QM = \Omega \s[\sqrt{\frac nk}]$\\
\hline
\end{tabular}
\]
\caption{Eight different options how the algorithm can use resources to solve the approximate counting problem.
Here, $\ell$ stands for the number of the copies of the state; while $\QM$, $\QG$, and $\QR$ stand for the number of executions of the membership, state-generating, and the reflecting oracles, respectively.
The rows without a bullet correspond to the case when the algorithm uses one type of resources.  The ones with a bullet correspond to pairs of resources.
There are two different options for the case of the reflecting oracle combined with the copies of the state or the state-generating oracle.
}
\label{tbl:main}
\end{table}

\begin{thm}
\label{thm:main}
Consider a quantum algorithm that distinguishes whether the input set $x\subseteq [n]$ has size $k$ or $k'=(1+\eps)k$.
For simplicity, we assume that $n\ge 5k$ and $1/k\le\eps\le 1$.
Suppose the algorithm uses $\ell$ copies of the state $\psi_x$ from~\rf{eqn:psix} and executes the membership, the state-generating, and the reflecting oracles or their inverses $\QM$, $\QG$, and $\QR$ times, respectively.
Then, in order to solve the problem, the resource consumption of the algorithm has to satisfy at least one of the eight conditions in \rf{tbl:main}.
\end{thm}

This is tight as demonstrated by the algorithms in \rf{app:upper}.
Let us note though that there is a $\log k$ discrepancy for the first term of Row 1 of the table, in the sense that we only show an $\ell = O(k\log k)$ algorithm in \rf{prp:coupon}.
However, similar techniques were used in a follow-up paper~\cite{belovs:couponCollector} to show an $\Omega(k\log k)$ lower bound for the task of learning the subset $x$.
All other lower bounds of \rf{thm:main} are tight up to constant factors.

\subsection{Techniques}
\label{sec:techniques}

A substantial difference between our paper and~\cite{aaronson:counting} is in the techniques used.  Aaronson \etal use the method of Laurent polynomials, whereas we use the adversary method combined with representation theory of the symmetric group.

The method of Laurent polynomials is a generalisation of the aforementioned polynomial method.
In Laurent polynomials, negative powers of the variables are allowed.
The positive degree bounds the number of executions of the membership oracle (like in the original version),  executions of the reflecting oracle, and the number of copies, whereas the negative degree bounds the latter two.

As mentioned previously, we use the general version of the adversary bound from~\cite{belovs:variations}.
Moreover, we use the version of the bound that allows for several input oracles that can be queried independently.
The adversary method with several input oracles was already used by Kimmel, Lin, and Lin in~\cite{kimmel:oraclesWithCosts}, but we use a different approach.

Let us return to the proof of \rf{thm:main}.
Using the above machinery, we formulate the lower bound as an optimisation problem.
Since the problem is symmetric, we can use representation theory of the symmetric group.
This gives a simplified optimization problem, which yields the required lower bound in a relatively simple way.
This simplification step can be applied to other symmetric problems that use access to the state~\rf{eqn:psix}.
Also, we give a different approach to representation theory of the symmetric group compared to previous papers like~\cite{rosmanis:smallRange}.
It is based on the use of group algebra to identify isotypical subspaces within a module.

Our technique has a number of advantages compared to the Laurent polynomials.
The Laurent-polynomial method depends crucially on the state being a uniform superposition over some set, the size of the set being an important parameter.
While Laurent polynomials can handle a restricted number of possible input resources,
adversary method, in principle, can be applied to any state-generating, reflecting, or any other type of input oracle.
Laurent polynomials also cannot distinguish between copies of the state and invocations of the reflecting oracle, whereas we are able to count all input resources independently.

\subsection{Organisation}
The paper can be divided into two big parts.
The first one, consisting of Sections~\ref{sec:preliminaries}--\ref{sec:lower}, covers quantum-computational aspect of the paper.
The second one, spanning over Sections~\ref{sec:JohnsonPrelim}--\ref{sec:bigModule}, deals with the representation-theoretical machinery.
Both parts are essentially independent from each other.
In particular, the first part can be read without any representation-theoretical background.  Similarly, the second part does not require any knowledge of quantum algorithms.

The gateway between the two parts are the main technical Lemmata~\ref{lem:CoeffsUnderA}--\ref{lem:DeltaINorm} in \rf{sec:adversaryGeneral}.
They characterise complexity with relation to various input resources as outlined above: copies of the state; state-generating, reflecting, and the membership oracles.
These lemmata are proven in the second part of the paper.

In \rf{sec:preliminaries}, we introduce the main notions from quantum computation.  We define the state conversion and the function evaluation problems, formally define the input oracles mentioned above, and also explain how we allow access to various input oracles simultaneously.

\rf{sec:adv} defines the general adversary bound, which is the main technical tool behind our lower bound.  In essence, this section parallels \rf{sec:preliminaries} and shows how the corresponding algorithmic objects look from the lower bound perspective.
The proofs of relevant results from~\cite{belovs:variations} and the related papers are given for completeness in the appendix.

\rf{sec:lower} is the proof of \rf{thm:main}.  First, we use the results of \rf{sec:adv} to formulate the optimisation problem yielding the lower bound for the approximate counting problem.  Next, the main technical Lemmata~\ref{lem:CoeffsUnderA}--\ref{lem:DeltaINorm} are stated, from which the proof of the lower bound is obtained.

\rf{sec:JohnsonPrelim} starts the second part of the paper.  It contains the main results from representation theory of the symmetric group required for the proofs of Lemmata~\ref{lem:CoeffsUnderA}--\ref{lem:DeltaINorm}.
It also gives a representation-theoretical analysis of the $S_n$-module $\bC^{\binom{[n]}k}$.

In \rf{sec:DeltaINormProof}, we prove the main lemma related to the membership oracle.
In \rf{sec:FirstProofs}, we prove the main lemmata related to the remaining types of resources.
They are related to the state $\psi_x$ from~\rf{eqn:psix} and are consequently based on representation-theoretical analysis of the $S_n$-module $\bC^{\binom{[n]}{k}}\otimes \bC^n$, which we perform in \rf{sec:bigModule}. 

The aim of \rf{app:upper} is to show that \rf{thm:main} is tight by providing the matching upper bounds.  Most of the algorithms are rather simple, but not without an occasional twist.
In Appendices~\ref{app:proofs} and~\ref{app:multipleProof}, we give proofs omitted from \rf{sec:adv}.
In \rf{app:posWeighted}, we argue that a natural generalisation of positive-weighted adversary cannot be applied to this version of the counting problem.

\subsection{Discussion and Future Work}

With Lemmata~\ref{lem:CoeffsUnderA}--\ref{lem:DeltaINorm} at our disposal, the proof of our lower bound is relatively concise and direct.
We believe that similar estimates can be helpful for other problems involving oracles that prepare uniform superpositions over some unstructured set.
Our approach to representation theory of the symmetric group follow a clear pattern, and can be used to analyse other related modules.

Let us mention one open problem here.
In the $k$-fold search problem, the input $x$ is a set of size $k$, and the task is to output $x$.
We are interested in trade-offs between different resources, as defined in this paper, \emph{and the success probability}.
This problem is solved for the membership oracle by Klauck, \v Spalek, and de Wolf~\cite{klauck:kFoldedSearch} using the polynomial method, and by Ambainis~\cite{ambainis:kFoldedSearch} using a variant of the adversary method.
A recent paper~\cite{belovs:couponCollector} is a preliminary step in that direction.

One complication is that our techniques are not well-suited for problems with small success probability.
We are using a variant of the so-called \emph{additive} adversary, which is not capable of proving better-than-linear dependence on the success probability~\cite{spalek:multiplicative}.
Therefore, one might require generalisation of the techniques in~\cite{belovs:variations} to allow small success probability and reformulation of the techniques of the current paper to fit that generalisation.
Let us note that the ideas from the aforementioned proof by Ambainis~\cite{ambainis:kFoldedSearch} were later generalised by \v Spalek into a \emph{multiplicative} version of the adversary bound~\cite{spalek:multiplicative}, suited for small success probabilities.
It is not clear at the moment how to merge the ideas of \v Spalek with our techniques.

Finally, the version of the adversary bound for several input oracles used in this paper inspired formulation of quantum Las Vegas query complexity in~\cite{belovs:LasVegas, belovs:transducers}.

\section{Preliminaries}
\label{sec:preliminaries}

We mostly use standard linear-algebraic notation.
The vector spaces are finite-dimensional complex inner product spaces.
We use ket-notation for vectors representing quantum states, but generally avoid it.
We use $A^*$ to denote conjugate operators (transposed and complex-conjugated matrices).
We use $A\elem[x,y]$ to denote the $(x,y)$-th entry of the matrix $A$.
Symbol $\oplus$ stands for direct sum of spaces, matrices, or vectors.
The context should suffice to tell them apart.

For matrices $A$ and $B$ of the same dimension, we use $A\circ B$ to denote their Hadamard (entry-wise) product.
We use $A^{\otimes \ell}$ and $A^{\circ\ell}$ to denote the $\ell$-th tensor and Hadamard powers, $A\otimes A\otimes\cdots \otimes A$ and $A\circ A\circ\cdots\circ A$, respectively, each repeated $\ell$ times.

For $P$ a predicate, we use $1_P$ to denote 1 if $P$ is true, and 0 if $P$ is false.
We use $[n]$ to denote the set $\{1,2,\dots,n\}$.
\medskip

In the remaining part of this section, we describe our model of quantum algorithms, few aspects of which are not standard.
In \rf{sec:model}, we define a quantum query algorithm and explain how it is supposed to solve a task.
In \rf{sec:multiple}, we discuss how we allow the algorithm to access several input oracles.
Finally, in \rf{sec:inputOracles}, we formally define the input oracles available to the algorithm.

\subsection{State Conversion and Function Evaluation}
\label{sec:model}

We first define a quantum query algorithm. For now, we consider the case when the algorithm has access to one input oracle.
Later, in \rf{sec:multiple}, we will explain how to deal with several input oracles.

A quantum query algorithm $A$ works in some space $\cH$, which we call the \emph{workspace} of the algorithm.
The algorithm is given an \emph{oracle} $O$, which is a unitary in some space $\cM$.
The interaction between the algorithm and the oracle is in the form of queries defined as follows.
We assume the space $\cM$ is embedded in $\cH$ in the following way:
$\cH = \cH_0\oplus (\cH_1\otimes \cM)$, and the \emph{query} is
\begin{equation}
\label{eqn:query}
\tO = I' \oplus (I\otimes O),
\end{equation}
where $I'$ and $I$ are the identities in $\cH_0$ and $\cH_1$, respectively.

\begin{defn}
\label{defn:algorithm}
A \emph{quantum query algorithm} $A$ with input oracle $O$ is a sequence of linear transformations of the following form:
\begin{equation}
\label{eqn:preAlgorithm}
A(O) = U_Q\, \tO^{\pm1} \, U_{Q-1}\,\tO^{\pm1}\,\cdots U_{1}\, \tO^{\pm1}\, U_0,
\end{equation}
where $U_t$ are some input-independent unitaries in $\cH$,
and each $\tO^{\pm1}$ is either a query $\tO$ or its inverse $\tO^*$.
\end{defn}

Now we define the \emph{state conversion} problem~\cite{lee:stateConversion} with \emph{general input oracles}~\cite{belovs:variations}.

\begin{defn}
\label{defn:stateConversion}
A \emph{state conversion problem} is specified by pairs of states $\xi_x\mapsto \tau_x$ and unitaries $O_x$, where $x$ ranges over some set $D$.
We assume here that $\xi_x \in \cK_0$, $\tau_x \in \cK_1$, and $O_x\colon \cM\to\cM$ for some spaces $\cK_0, \cK_1$ and $\cM$.
We say that an algorithm $A$ as in \rf{defn:algorithm} \emph{solves} this problem if, assuming that $\cK_0$ and $\cK_1$ are embedded in $\cH$, we have $A(O_x)\xi_x = \tau_x$ for all $x\in D$.
The choice of embedding is irrelevant.
\end{defn}

In the settings of \rf{defn:stateConversion}, we will usually write $O = (O_x)$ to denote the collection of oracles over all $x\in D$.

The most common case of state conversion is evaluation of a function $f$ defined on some set $D$.
In this case (assuming exact computation for now), one can let $\xi_x = \ket|0>$ and $\tau_x = \ket|f(x)>$.
In our version of approximate counting, the goal is to evaluate a function as well: $f(x)=0$ if $|x|=k$, and $f(x)=1$ if $|x|=k'$.
However, the algorithm is also provided with an additional resource: a state $\xi_x$ that depends on $x\in D$, consisting of multiple copies of the state $\psi_x$.
And, as usual, we allow the algorithm to err with small probability.
This gives the following definition.
\begin{defn}[Boolean Function Evaluation with Initial States]
\label{defn:functionEvaluation}
Let $f\colon D\to\bool$ be a function, and $\cM$ and $\cK_0$ be vector spaces.
For each $x\in D$, let $\xi_x\in\cK_0$ be a quantum state and $O_x$ be a unitary acting in $\cM$.
We say that a quantum algorithm \emph{evaluates the function $f$ with initial states $\xi_x$ and permitted error $\delta$} if it solves the state conversion problem $\xi_x\mapsto\tau_x$ for some collection of vectors $\tau_x\in \cH = \bC^2\otimes \cH'$ such that measuring $\tau_x$ in the first register gives $f(x)$ with probability at least $1-\delta$ for all $x\in D$.
\end{defn}

\begin{rem}[Error Reduction]
\label{rem:reduce}
In \rf{defn:functionEvaluation}, we opted out of having a fixed value of $\delta$, say $1/3$, because, in general, it is not possible to perform standard error reduction by repetition for this task.
The reason is that the algorithm is given only a single copy of the initial state $\xi_x$.
However, for our specific problem of approximate counting, this issue is irrelevant, because the state $\xi_x$ is defined as $\psi_x^{\otimes \ell}$ for some $\ell$.
Therefore, having multiple copies of $\xi_x$ results in changing $\ell$ by a constant factor, which we mostly allow.
In the occasional cases when we do not, like Row 7 of \rf{tbl:main} below, we solve it separately.
\end{rem}

\subsection{Multiple Input Oracles}
\label{sec:multiple}

Assume that in the settings of state conversion, \rf{defn:stateConversion}, the algorithm has not just one input oracle $O = (O_x)$, but several $O^{(1)} = \sA[O^{(1)}_x], \dots, O^{(r)} = \sA[O^{(r)}_x]$, all indexed by $x\in D$.
Each oracle $O^{(i)}$ acts in some space $\cM^{(i)}$.

It is possible to come up with various ways of providing access to these oracles simultaneously.
For instance, it is possible to query them in some pre-defined way.
It is also possible to perform intermediate measurements, which would classically determine which of the input oracles should be applied on each step.
This is the model we assume in \rf{app:upper} for the algorithms that match the lower bounds of \rf{thm:main}.

Since we are interested in proving lower bounds, we allow as flexible a model as possible, which later was formalised as quantum Las Vegas query complexity in~\cite{belovs:LasVegas}.
This model incorporates the scenarios as above via standard technique of postponed measurements.
Additionally, it is easy to reconcile with \rf{defn:stateConversion}, as we formally combine all the input oracles into one.
The only thing that changes is the way how we count the number of queries.

We model our algorithm as a usual single-oracle quantum query algorithm as in \rf{defn:stateConversion} with access to the combined input oracle $O_x$ on $\cM = \cM^{(1)}\oplus\cdots \oplus \cM^{(r)}$ given by
\[
O_x = O^{(1)}_x\oplus O^{(2)}_x \cdots \oplus O^{(r)}_x.
\]
We will not restrict the number of queries to this combined oracle per se.
Instead of that, we will account for the query complexity of each of the individual input oracles post factum.

Let $\psi_{t,x}$ be the state of the algorithm on input $x$ just before the $t$-th application of the combined input oracle.
As in \rf{defn:stateConversion}, assume that on the $t$-th query, the oracle is executed as
\[
I'\oplus (I\otimes O_x) 
= 
I' \oplus \sA[I\otimes O^{(1)}_x]\oplus \sA[I\otimes O^{(2)}_x] \cdots \oplus \sA[I\otimes O^{(r)}_x] .
\]
Decompose the state $\psi_{t,x}$ in the same way:
\[
\psi_{t,x} = \psi^{(0)}_{t,x} \oplus \psi^{(1)}_{t,x} \oplus \psi^{(2)}_{t,x} \oplus \cdots \oplus \psi^{(r)}_{t,x}
\]
so that $\psi^{(i)}_{t,x}$ is processed by $I\otimes O^{(i)}_x$, and $\psi^{(0)}_{t,x}$ remains intact.
In the case when the inverse oracle $I' \oplus (I\otimes O_x^{-1})$ is applied, we define $\psi^{(i)}_{t,x}$ similarly.

Now we \emph{define} the query complexity of the oracle $O^{(i)}$ on input $x$ as
\begin{equation}
\label{eqn:Ti_x}
L^{(i)}_x = \sum_t \norm|\psi_{t,x}^{(i)}|^2.
\end{equation}
Thus, query complexity can vary for different inputs $x$.
The total query complexity of the $i$-th input oracle is defined as the maximum of $L^{(i)}_x$ over all $x\in D$.


\subsection{Definitions of the Input Oracles}
\label{sec:inputOracles}
In this section, we define the input oracles available to the algorithm.
We use the same notation $O=(O_x)$ to denote any input oracle, as it will be clear from the context which one is meant in each case.
For state-generating and membership oracles, we first give a more general definition, and then a more restricted version, which is without loss of generality.

\begin{itemize}
\item 
First, we define the \emph{state-generating input oracle}.
In this case, for each $x\in D$, the algorithm is given a black-box access to a unitary $O_x$ performing the transformation $\ket|0> \mapsto \ket|\psi_x>$ in $\cM=\bC^n$, where $\ket|0>$ is some fixed state in $\cM$.
The algorithm should work for any unitary $O_x$ performing this transformation%
\footnote{%
\label{foot:stateGenerating}
To formally encapsulate this into the formalism of \rf{defn:stateConversion}, we can assume the labels are of the form $(x, O)$, where $x\in D$ is the original label, and $O$ is some legitimate input oracle for $x$.
}.

Alternatively, one may assume that $\cM$ is $n+1$-dimensional: $\cM = \bC^{\{0\}\cup[n]}$, and the state $\ket|0>$ is orthogonal to all $\psi_x$.
The transformation $O_x$ is the reflection about the orthogonal complement of $(\ket|0> - \ket|\psi_x>)/\sqrt2$:
\begin{equation}
\label{eqn:oracleStateGenerating}
O_x = I - \sA[\ket|0> - \ket|\psi_x>]\sA[\ket|0> - \ket|\psi_x>]^*.
\end{equation}
This is a special case of the above transformation, and it can be implemented using two invocations of \emph{any} input oracle mapping $\ket|0> \mapsto \ket|\psi_x>$: one inverse and one direct.  

\item
Second, we define the \emph{reflecting input oracle}.
In this case, the black-box unitary $O_x$ is the reflection about the state $\psi_x$: $O_x = 2\psi_x\psi_x^* - I$ in $\cM = \bC^n$.

\item
Third, we use the \emph{standard (membership) input oracle}.
In this case, we treat $D$ as a subset of the set $\bool^n$ of bit-strings of length $n$ and $\cM = \bC^n\otimes \bC^2$.
For $x\in D$, the corresponding input oracle performs the transformation $O_x\colon \ket|i>\ket|0>\mapsto \ket|i>\ket|x_i>$.
The algorithm should work for any unitary $O_x$ performing this transformation.
This can be seen as a direct sum of state-generating oracles performing transformations $\ket|0>\mapsto\ket|x_i>$ over all $i\in [n]$.

Alternatively, one may assume that $O_x$ performs the transformation 
\begin{equation}
\label{eqn:oracleMembership}
O_x\colon \ket|i>\ket|b> \mapsto \ket|i>\ket|b\oplus x_i>,
\end{equation}
where $i\in [n]$, $b\in\bool$ and $\oplus$ is the XOR operation.
This is the usually assumed special case of the above transformation~\cite{buhrman:querySurvey}.
And, again, it can be implemented using two invocations of any input oracle mapping $O_x\colon \ket|i>\ket|0>\mapsto \ket|i>\ket|x_i>$.
\end{itemize}

It is easy to see that, for any given state, the state-generating input oracle is at least as strong as both the reflecting oracle and a copy of the state.
Indeed, it is possible to implement the reflecting oracle using the state-generating oracle twice; and, using it once, it is possible to get a copy of the state.
Other than that, the above resources are incomparable.

\section{General Adversary Method}
\label{sec:adv}
The aim of this section is to describe the techniques used in the proof of the lower bound.
We extend on the general adversary method from~\cite{belovs:variations}, and we adopt a lower-bound-related perspective.

\subsection{Formulation for Multiple Input Oracles}

We need few pieces of notation.
First, for an input oracle $O = (O_x)_{x\in D}$, we define a family of matrices \emph{representing} the input oracle, where $x,y$ range over $D$:
\begin{equation}
\label{eqn:oracleObject}
\Delta_{x,y} = O_x - O_y.
\end{equation}

For pairs $\xi_x\mapsto \tau_x$ of states, where $x$ ranges over $D$, we define a $D\times D$-matrix $E$ given by
\begin{equation}
\label{eqn:stateConversionObject}
E\elem[x,y] = \ip<\xi_x, \xi_y> - \ip<\tau_x, \tau_y>,
\end{equation}
which \emph{represents} state conversion $\xi_x\mapsto \tau_x$.
We often treat $E$ as a family of $1\times 1$-matrices indexed by $x,y\in D$.

Assume $(A_{x,y})_{x,y\in D}$ is a family of matrices all of the same size.
Let $\Gamma$ be an $X\times Y$ matrix with $X,Y\subseteq D$.
We define $\Gamma\circ A$ as the $X\times Y$ block matrix, where the block corresponding to $x\in X$ and $y\in Y$ is given by $\Gamma\elem[x,y]A_{x,y}$.
This is a generalisation of the usual Hadamard product.

Now we define a variant of the adversary bound for multiple input oracles, which follows the same ideas as in the proof of negative-weighted adversary~\cite{hoyer:advNegative, belovs:variations}.

\begin{thm}[State Conversion, Multiple Input Oracles]
\label{thm:multipleGamma2}
Consider the state-conversion problem $\xi_x\mapsto \tau_x$ with input oracles $O^{(1)}, O^{(2)},\dots, O^{(r)}$.
Let  $\Delta^{(1)},\Delta^{(2)},\dots,\Delta^{(r)}$ be defined like in~\rf{eqn:oracleObject} for these oracles, and $E$ be as in~\rf{eqn:stateConversionObject}.
Then, for any $D\times D$-matrix $\Gamma$ and every algorithm performing the state conversion exactly, we have
\[
\sum_{i=1}^r \normA|\Gamma\circ\Delta^{(i)}| \max_{x\in D} L_x^{(i)} \ge \|\Gamma\circ E\|,
\]
where $L_x^{(i)}$ are as in~\rf{eqn:Ti_x}.
\end{thm}

We prove this theorem in \rf{app:multipleProof}.
The following easy corollary is useful.

\begin{cor}
\label{cor:multipleGamma2}
Consider the state-conversion problem $\xi_x\mapsto \tau_x$ with $r=O(1)$ input oracles $O^{(1)}, O^{(2)},\dots, O^{(r)}$.
Let $E$ be as in~\rf{eqn:stateConversionObject}, and  $\Delta^{(1)},\Delta^{(2)},\dots,\Delta^{(r)}$ be as in~\rf{eqn:oracleObject}.

Then, for any $D\times D$-matrix $\Gamma$, and every algorithm performing the state conversion exactly, there exists $i$ such that the query complexity of $O^{(i)}$ is 
\begin{equation}
\label{eqn:multipleGamma}
\max_{x\in D} L^{(i)}_x = \Omega\s[\frac{\|\Gamma\circ E\|}{\|\Gamma\circ \Delta^{(i)}\|}].
\end{equation}
\end{cor}

For one input oracle, and if $E$ corresponds to function evaluation, this becomes the usual adversary bound by H\o yer \etal~\cite{hoyer:advNegative}.
Intuitively, our formulation says that if we want to prove a lower bound against a collection of oracles, it suffices to find an adversary matrix $\Gamma$ which is a solution to all single-oracle versions of the bounds simultaneously.

\subsection{Alternative Input and Output Conditions}
\label{sec:advInputOracles}

\rf{cor:multipleGamma2} is rather cumbersome to apply in its original form due to various $\tau_x$ vectors allowed in \rf{defn:functionEvaluation} because of permitted error $\delta$, and the necessity to deal with complicated matrices $\Delta_{x,y}$.
Luckily, it is possible to simplify the bound using equivalent ways to represent input and output conditions.

We start with the output condition for approximate function evaluation with initial states.
We assume notation of \rf{defn:functionEvaluation}.
In particular, $\delta$ is the error parameter of the algorithm, and we also let $X=f^{-1}(1)$ and $Y = f^{-1}(0)$.
In this context, we will switch from $D\times D$-matrices to $X\times Y$-matrices.
Let $\Xi$ be the $X\times Y$-matrix defined by $\Xi\elem[x,y] = \ip<\xi_x, \xi_y>$.
The following result~\cite{lee:stateConversion} effectively reduces the approximate case to the exact case.
See \rf{app:proofs} for a proof.

\begin{lem}
\label{lem:advFunctionEvaluation}
In the above notation, assume that $\Gamma$ is an $X\times Y$-matrix that satisfies
\begin{equation}
\label{eqn:advFunctionCondition}
\|\Gamma\| = 1
\qqand
\|\Gamma\circ\Xi\| \ge 3\sqrt\delta.
\end{equation}
Then, for every state conversion problem $\xi_x\mapsto \tau_x$ that corresponds to evaluation of the function $f$ with initial states $\xi_x$ and permitted error $\delta$, we get
\begin{equation}
\label{eqn:advFunctionOutcome}
\|\Gamma\circ E\| \ge \sqrt\delta,
\end{equation}
where $E$ is defined in~\rf{eqn:stateConversionObject}.
\end{lem}

\mycutecommand{\Dp}{\Delta^{\psi}}
\mycutecommand{\Dps}{\Delta^{\psi^*}}
\mycutecommand{\Dpp}{\Delta^{\psi\psi^*}}

Now let us move to the input conditions.
The key idea is the following definition.

\begin{defn}[$\gamma_2$-equivalence]
\label{defn:equivalence}
Let $\Delta = (\Delta_{x,y})$ and $\Delta' = (\Delta'_{x,y})$ be two families of matrices both indexed by $x,y\in D$.
If, for every $D\times D$-matrix $\Gamma$, we have
\[
\frac{\|\Gamma\circ \Delta\|}{\|\Gamma\circ\Delta'\|} = \Theta(1),
\]
then we say that $\Delta$ and $\Delta'$ are \emph{$\gamma_2$-equivalent}.
If $\Delta'$ is $\gamma_2$-equivalent to the family $\Delta$ defined in~\rf{eqn:oracleObject}, then we say that $\Delta'$ \emph{represents} the input oracle $O = (O_x)$.
\end{defn}

We will often have that $\Delta' = A^{(1)}\oplus\cdots\oplus A^{(\ell)}$ meaning that $\Delta'_{x,y} = A^{(1)}_{x,y}\oplus\cdots\oplus A^{(\ell)}_{x,y}$ for all $x,y\in D$.
In this case, we will say that the input oracle is represented by a tuple $A^{(1)},\dots,A^{(\ell)}$.
The following identity can be used in this case:
\[
\normA|\Gamma\circ (A^{(1)}\oplus\cdots\oplus A^{(\ell)})| 
= \max_{i\in[\ell]}\normA|\Gamma\circ A^{(i)}|.
\]

The idea is that we replace $\Delta^{(i)}$ in the denominator of~\rf{eqn:multipleGamma} with  any $\gamma_2$-equivalent $\Delta^{\prime(i)}$, which changes the bound by at most a constant factor.
Thus, we can find a representation that is easier to work with.
This is what we will do in the remaining part of this section for the input oracles defined in \rf{sec:inputOracles}.
The results are from previous papers, and their proofs are given in~\rf{app:proofs} for completeness.

For the reflecting oracle, defined by $O_x = 2\psi_x\psi_x^*-I$, we will use the standard definition~\rf{eqn:oracleObject} removing the constant factor 2.
Thus, the following family represents the reflecting oracle:
\begin{equation}
\label{eqn:Deltapsipsi}
\Dpp_{x,y} := \psi_x\psi_x^* - \psi_y\psi_y^*.
\end{equation}
We use mnemonic $\psi\psi^*$, which reminds of the definition.

The state-generating oracle is represented~\cite{belovs:distributions} by the family $\Dp\oplus\Dps$, where 
\begin{equation}
\label{eqn:Deltapsi}
\Dp_{x,y} := \psi_x - \psi_y
\qqand
\Dps_{x,y} := \psi_x^* - \psi_y^*.
\end{equation}

\mycutecommand{\Dmem}{\Delta^{\mathrm{membership}}}

Finally, the following family represents~\cite{belovs:variations} the membership oracle:
\begin{equation}
\label{eqn:Delta_xy}
\Dmem_{x,y} := \bigoplus\nolimits_{i\in [n]} 1_{x_i\ne y_i}.
\end{equation}
Here $\bigoplus$ stands for the direct sum of $1\times 1$ matrices, hence, $\Dmem_{x,y}$ is a diagonal $n\times n$ matrix.
The usual way to write $\Dmem$ is as $\Delta_1\oplus\cdots\oplus \Delta_n$, where
\begin{equation}
\label{eqn:Deltai}
\Delta_i\elem[x,y] := 1_{x_i\ne y_i}.
\end{equation}
Again, we treat these $D\times D$-matrices as families of $1\times 1$-matrices.

\subsection{Application to Approximate Counting}
\label{sec:technicalFormulation}

Let us apply the above results for the special case of approximate counting of \rf{thm:main}.
In this case, the set $D$ of labels is $X \cup Y$, where $X$ consists of all the $n$-bit strings of Hamming weight $k$ and $Y$ consists of all the $n$-bit strings of Hamming weight $k' = (1+\eps)k$.

Define an $X\times Y$ matrix $\Psi$ by
\begin{equation}
\label{eqn:Psi}
\Psi\elem[x,y] := \ip<\psi_x, \psi_y>,
\end{equation}
where $\psi_x$ are defined in~\rf{eqn:psix}.
The Gram matrix of the initial states, $\Xi$ of~\rf{lem:advFunctionEvaluation}, is given by
\begin{equation}
\label{eqn:XiSpecific}
\Xi = \Psi^{\circ\ell},
\end{equation}
where $\ell$ is the number of copies of the state $\psi_x$ available to the algorithm.

\begin{prp}
\label{prp:formulation}
Assume we can find an $X\times Y$ matrix $\Gamma$ such that 
\itemstart
\item $\|\Gamma\| = 1$;
\item $\|\Gamma\circ \Psi^{\circ\ell}\| = \Omega(1)$ with $\Psi$ as in~\rf{eqn:Psi};
\item for each $i\in[n]$, we have $\|\Gamma\circ\Delta_i\| \le 1/T_1$ with $\Delta_i$ as in~\rf{eqn:Deltai};
\item both $\|\Gamma\circ\Dp\|$ and $\|\Gamma\circ\Dps\|$ are at most $1/T_2$ with $\Dp$ and $\Dps$ as in~\rf{eqn:Deltapsi};
\item and $\|\Gamma\circ\Dpp\| \le 1/T_3$ with $\Dpp$ as in~\rf{eqn:Deltapsipsi}.
\itemend
Then, there exists a constant $C>0$ such that every quantum algorithm, given $C\ell$ copies of the state $\psi_x$, and solving the approximate counting problem with bounded error must
\itemstart
\item either execute the membership oracle $\Omega(T_1)$ times;
\item or execute the state-generating oracle $\Omega(T_2)$ times;
\item or execute the reflecting oracle $\Omega(T_3)$ times.
\itemend
\end{prp}

In order to prove~\rf{thm:main} it suffices to find a feasible solution $\Gamma$ for this optimisation problem with good objective value.

\pfstart[Proof of \rf{prp:formulation}]
Let a constant $C'$ be such that $\|\Gamma\circ \Psi^{\circ\ell}\|\ge C'$.
Assume the error parameter of the algorithm is $\delta=(C'/3)^2 = \Omega(1)$.
We may achieve this by repeating the algorithm $1/C$ times for some constant $C$.
By \rf{rem:reduce}, we can do this, assuming we have $\frac1C\cdot C\ell = \ell$ copies of the state $\psi_x$.

The proposition follows from \rf{cor:multipleGamma2} using the representations of the input oracles from \rf{sec:advInputOracles}, since by \rf{lem:advFunctionEvaluation} and~\rf{eqn:XiSpecific}, $\|\Gamma\circ E\| = \Omega(1)$.
\pfend

\section{Lower Bound}
\label{sec:lower}
In this section, we prove the lower bound for approximate counting, \rf{thm:main}.
First, in \rf{sec:adversaryGeneral}, we formulate the general form of the adversary matrix $\Gamma$, which follows from representation theory, and state the main technical results: Lemmata~\ref{lem:CoeffsUnderA}--\ref{lem:DeltaINorm}.  
Their proofs constitute the second part of this paper.
In \rf{sec:construction}, we fix parameters in this general form, and get specific estimates.
Finally, in \rf{sec:proof}, we use these results to prove \rf{thm:main}.

\subsection{Main Lemmata}
\label{sec:adversaryGeneral}

In this section, we give the general form of the adversary matrix $\Gamma$ and give estimates on the norm of the matrix when various $\Delta$-operators are applied to it.
Using a variant of the automorphism principle as in~\cite{hoyer:advNegative}, we can assume that the matrix $\Gamma$ is symmetric with respect to the permutations of the input variables.
Then representation theory of the symmetric group tells us that the adversary matrix has the form
\begin{equation}
\label{eqn:Gamma}
\Gamma = \sum_{j=0}^k \gamma_j \Phi_j,
\end{equation}
where $\Phi_j$ are the isometric isomorphisms between the copies of the irreps of the symmetric group in $S_n$-modules $\bC^{\binom {[n]}{k'}}$ and $\bC^{\binom {[n]}k}$.
(See \rf{sec:JohnsonPrelim} for more detail.)
The matrices $\Phi_j$ will only have real entries though, and all $\gamma_i$ will be non-negative real numbers.
For the purposes of this section, it suffices to know that the ranges of different $\Phi_j$ are pairwise orthogonal, and so are their coimages.
The isometry implies $\|\Phi_j\|=1$ for all $j$.
Hence, in particular, $\|\Gamma\| = \max_j |\gamma_j|$.

\mycutecommand{\fita}{\phi}

Now let us describe how operations from~\rf{sec:advInputOracles} act on matrices of the form~\rf{eqn:Gamma}.
Here we only state the results, the proofs being postponed to the next sections.
We will need the following $4\times 4$ diagonal matrices and 4-dimensional vectors, where $j$ ranges from 0 to $k$.
\begin{equation}
\label{eqn:fita}
\Gamma_j = \begin{pmatrix}
\gamma_{j-1} & & &\\
& \gamma_j & & \\
&& \gamma_j &\\
&&& \gamma_{j+1}
\end{pmatrix}
\qqand
\renewcommand{\arraystretch}{1.8}
\fita^k_j = 
\begin{pmatrix}
\fita_{j,1}^k\\
\fita_{j,2}^k\\
\fita_{j,3}^k\\
\fita_{j,4}^k
\end{pmatrix}
=
\begin{pmatrix}
\sqrt{\frac{j(k-j+1)(n-k-j+1)}{(n-2j+2)(n-2j+1)k}}\\
\sqrt{\frac{k}{n}}\\
\frac{n-2k}{\sqrt{nk}} \sqrt{\frac{j(n-j+1)}{(n-2j+2)(n-2j)}}\\
\sqrt{\frac{(n-j+1)(k-j)(n-k-j)}{(n-2j+1)(n-2j)k}}
\end{pmatrix}.
\end{equation}
In $\Gamma_0$, it is assumed that $\gamma_{-1} = 0$.  By our assumption on $n\ge 5k$, all the entries of $\fita^k_j$ are non-negative real numbers.

\begin{clm}
\label{clm:phi_j is unit}
The vector $\fita_j^k$ has norm 1.
\end{clm}

\pfstart
This can be verified by a direct computation. 
However, there is a reason behind this, since the entries of $\fita_j^k$ give the coordinates of a unit vector with respect to an orthonormal basis, see \rf{sec:operatorV}.
\pfend

The following lemmata will be proven in \rf{sec:FirstProofs}.
We assume here that $\Gamma$ is an $X\times Y = \binom{[n]}{k}\times \binom{[n]}{k'}$ matrix defined as in~\rf{eqn:Gamma}, and $\fita_j = \fita_j^k$ and $\fita'_j = \fita_j^{k'}$ as above.

\begin{lem}
\label{lem:CoeffsUnderA}
For $\Gamma$ as in~\rf{eqn:Gamma} and $\Psi$ as in~\rf{eqn:Psi}, we have
\[
\Gamma\circ \Psi = \sum_{j=0}^k  \fita_j^* \Gamma_j \fita'_j\cdot \Phi_j. 
\]
\end{lem}

\begin{lem}
\label{lem:DeltaGenNorm}
For $\Gamma$ as in~\rf{eqn:Gamma} and $\Dp$ and $\Dps$ as in~\rf{eqn:Deltapsi}, we have
\[
\norm|\Gamma\circ\Dp| = 
\max_j 
\norm|\gamma_j \fita_j - \Gamma_j \fita'_j |
\qqand
\norm|\Gamma\circ\Dps| = 
\max_j 
\norm|\Gamma_j \fita_j - \gamma_j \fita'_j |.
\]
\end{lem}

\begin{lem}
\label{lem:DeltaReflNorm}
For $\Gamma$ as in~\rf{eqn:Gamma} and $\Dpp$ as in~\rf{eqn:Deltapsipsi}, we have
\[
\norm|\Gamma\circ\Dpp| = \max_j \norm| \Gamma_j \fita_j \fita_j^* - \fita_j' {\fita_j'}^*\Gamma_j |.
\]
\end{lem}

The following lemma is proven in~\rf{sec:DeltaINormProof}.

\begin{lem}
\label{lem:DeltaINorm}
For $\Gamma$ as in~\rf{eqn:Gamma}, $\Delta_i$ as in~\rf{eqn:Deltai}, and for all values of $i\in[n]$, we have
\begin{align*}
\|\Gamma\circ\Delta_i\| =
\max_j \max \bigg\{\,
&\bigg|\frac{\sqrt{(k-j)(n- k'-j)}}{n-2j} \gamma_j - \frac{\sqrt{(k'-j)(n- k-j)}}{n-2j} \gamma_{j+1}\bigg|,\\
&\bigg|\frac{\sqrt{(k'-j)(n- k-j)}}{n-2j} \gamma_j - \frac{\sqrt{(k-j)(n- k'-j)}}{n-2j} \gamma_{j+1}\bigg|\,\bigg\}.
\end{align*}
\end{lem}

We end this section with the following estimates on the entries of $\fita_j' - \fita_j$.

\begin{lem}
\label{lem:fita}
Assume that $n\ge 5k$ and $k\ge 5j$.
For $\fita_j = \fita_j^k$ as defined in~\rf{eqn:fita} and $\fita'_j = \fita_j^{k'}$ with $k' = (1+\eps)k \le 2k$, we have the following estimates on the entries of $\fita_j$, $\fita_j'$ and their differences:
\[
\begin{array}{|c|cccc|}
\hline
i & 1 & 2 & 3 & 4\\ \hline
\fita_{j,i},\;\fita'_{j,i} & 
\displaystyle O\sC[\sqrt{\frac jn}] &
\displaystyle O\sC[\sqrt{\frac kn} ] &
\displaystyle O\sC[\sqrt{\frac jk}] &
\displaystyle O(1)
\\
|\fita_{j,i} - \fita_{j,i}'| & 
\displaystyle O\sC[\eps \frac{j^{3/2} }{k\sqrt n} + \eps\frac{\sqrt jk}{n^{3/2}} ] &
\displaystyle O\sC[\eps\sqrt{\frac kn}] &
\displaystyle  O\sC[\eps\sqrt{\frac jk}] &
\displaystyle O\sC[ \eps \frac jk + \eps \frac kn ]\\
\hline
\end{array}
\]
In particular, where $\|\cdot\|_1$ stands for the $\ell_1$-norm:
\begin{equation}
\label{eqn:phiDifferenceEstimate}
\normA| \fita_j - \fita_j' | \le \normA| \fita_j - \fita_j'|_1 
= O\sC[\eps \sqrt{\frac kn} + \eps \sqrt{\frac jk} ].
\end{equation}
\end{lem}

\pfstart
The estimates on the values are straightforward.
The estimates on the differences are also easy to derive using the inequality
$
\abs|h(1+\eps) - h(1)| \le \int_{1}^{1+\eps} |h'(t)|\, \mathrm{d}t
$.

For $|\fita'_{j,1}-\fita_{j,1}|$, we have the upper bound of
\[
O\s[\frac{\sqrt j}{n}] 
\abs|\sqrt{\s[1-\frac {j+1}{k(1+\eps)}]\sA[n-j+1-k(1+\eps)]} - \sqrt{\s[1-\frac {j+1}k](n-j+1-k)} |,
\]
which is $O\s[\eps\sqrt{j/n}](j/k + k/n)$.
For $|\fita'_{j,2}-\fita_{j,2}|$, we have
\[
\abs| \sqrt{\frac{k(1+\eps)}n} - \sqrt{\frac kn} | = 
O\s[\eps\sqrt{\frac kn}].
\]
For $|\fita'_{j,3}-\fita_{j,3}|$, we have
\[
O\s[\frac{\sqrt j}n]
\abs| \frac{n-2k(1+\eps)}{\sqrt{k(1+\eps)}} - \frac{n-2k}{\sqrt{k}} | 
= O\sC[\eps \frac{\sqrt{jk}}n + \eps\sqrt{\frac jk}] = O\sC[\eps\sqrt{\frac jk}].
\]
For $|\fita'_{j,4}-\fita_{j,4}|$, we have
\[
O\s[\frac1{\sqrt{n}}]
\abs|\sqrt{\s[1-\frac j{k(1+\eps)}]\sA[n-k(1+\eps)-j]} - \sqrt{\s[1-\frac jk](n-k-j)} |
= O\sC[ \eps \frac jk + \eps \frac kn ].\qedhere
\]
Equation~\rf{eqn:phiDifferenceEstimate} follows from the last row of the table and by removing the subdominant terms.
\pfend

\subsection{Specific Form and Estimates}
\label{sec:construction}

It turns out, that for all the lower bounds of \rf{thm:main}, it suffices to have a very simple choice of the parameters $\gamma_j$ in~\rf{eqn:Gamma}.
Our choice follows a rather standard construction~\cite{belovs:kSumLower, belovs:variations} of
a gradient between a high value of the coefficient for $j=0$ and 0 for large $j$.
The adversary matrix is specified by a single parameter $1\le t\le k/5$, which dictates the steepness of the gradient.
By changing $t$, we will be able to ``slide'' along the trade-off curve.
Namely, we define
\begin{equation}
\label{eqn:gammaj}
\gamma_j = \max\sfig{1 - \frac jt, 0  }.
\end{equation}
We have $\|\Gamma\|=1$.  
Also, $\gamma_i=0$ for all $j\ge t$.

We will now use the general results from \rf{sec:adversaryGeneral} to get estimates in this specific case.


\begin{prp}
\label{prp:gammaNormPsi}
In order to have $\|\Gamma\circ \Psi^{\circ\ell}\| = \Omega(1)$, it suffices that the following conditions are met:
\[
t \ge 2\ell,
\qqand
\ell \le C \min\sfigC{\frac{\sqrt k}{\eps},\;  \frac{n}{k\eps^2}}
\]
for a sufficiently small constant $C$.
\end{prp}

\pfstart
Using \rf{lem:CoeffsUnderA} repeatedly, we see that 
$\Gamma\circ \Psi^{\circ s} = \sum_{j=0}^k\gamma_j^{(s)} \Phi_j$ for some coefficients $\gamma_j^{(s)}$.
For $s \in \{0,1,\dots,\ell\}$, let us denote
\[
D = \min\nolimits_{j=0}^\ell \ipO<\fita_j, \fita_j'>
\qqand
\gamma^{(s)} = \min\nolimits_{j=0} ^{\ell-s} \gamma^{(s)}_j.
\]
First, from our assumption $t\ge 2\ell$ and using~\rf{eqn:gammaj}, we get that $\gamma^{(0)} \ge 1/2$.
Next, from \rf{lem:CoeffsUnderA} applied to $\Gamma\circ\Psi^{\circ s}$, we get for $s<\ell$ and $j < \ell-s$:
\[
\gamma_j^{(s+1)}  = 
\gamma^{(s)}_{j-1}\fita_{j,1}\fita'_{j,1} + 
\gamma^{(s)}_{j}\fita_{j,2}\fita'_{j,2} +
\gamma^{(s)}_{j}\fita_{j,3}\fita'_{j,3} +
\gamma^{(s)}_{j+1}\fita_{j,4}\fita'_{j,4} 
\ge \gamma^{(s)} \ip<\fita_j, \fita_j'> \ge \gamma^{(s)} D.
\]
This implies by induction $\gamma^{(\ell)} \ge D^{\ell}/2$, and
\[
\|\Gamma\circ \Psi^{\circ \ell}\| \ge \gamma_0^{(\ell)} = \gamma^{(\ell)} \ge D^\ell/2.
\]

To prove the lemma, it suffices to lower bound $D$.
Let $\alpha_j$ be the angle between $\fita_j$ and $\fita_j'$.  
Since both vectors have unit norm and using~\rf{eqn:phiDifferenceEstimate}:
\[
\ipO<\fita_j, \fita_j'>^2 = 
\cos^2 \alpha_j = 
1 - \sin^2 \alpha_j  =  1 - O\s[\| \fita_j - \fita_j' \|^2 ]
= 1 - O\s[\eps^2 \frac jk + \eps^2\frac kn].
\]
Hence, by the definition of $D$:
\[
D^\ell 
\ge \sC[{1 - O\s[\frac \ell k\eps^2 + \frac kn\eps^2]}]^{\ell/2}
\ge 1 - O\s[\frac{\ell^2}{k} \eps^2 + \frac {\ell k}n\eps^2 ] = \Omega(1)
\]
by our assumption on $\ell$.
This ends the proof of \rf{prp:gammaNormPsi}.
\pfend

\begin{prp}
Both $\norm|\Gamma\circ\Dp|$ and $\norm|\Gamma\circ\Dps|$ are 
\( \displaystyle
O\sC[ \eps\sqrt{\frac kn} + \eps \sqrt{\frac tk} + \frac 1t ].
\)
\end{prp}

\pfstart
We prove the estimate on $\norm|\Gamma\circ\Dps|$, the second one being completely analogous.
We use \rf{lem:DeltaGenNorm}.
If $j>t+1$, we get $\Gamma_j \fita_j - \gamma_j \fita'_j = 0$ by the definition of $\gamma_j$ in~\rf{eqn:gammaj}.
On the other hand, for a fixed $j\le t+1$: 
\[
\norm|\Gamma_j \fita_j - \gamma_j \fita'_j | 
\le 
\absA|\gamma_j-\gamma_{j-1}|\fita_{j,1} +
\absA|\gamma_j-\gamma_{j+1}|\fita_{j,4} +
\gamma_j\|\fita_j' - \fita_j\|
= O\sC[ \frac1t + \eps \sqrt{\frac kn} + \eps \sqrt{\frac jk} ],
\]
where we used that $|\gamma_j-\gamma_{j-1}| \le 1/t$, all $\fita_{j,i}$ and $\gamma_j$ are at most 1, and~\rf{eqn:phiDifferenceEstimate}.
The lemma follows by our assumption $j\le t+1$.
\pfend

\begin{prp}
We have $\displaystyle \|\Gamma\circ\Dpp\| 
= O\sC[\frac 1t + \eps]\sC[\sqrt{\frac kn} + \sqrt{\frac tk}]$.
\end{prp}

\pfstart
We have
\begin{align*}
\Gamma_j{ \fita_j} \fita_j^*
&=
\begin{pmatrix}
\gamma_{j-1}\fita_{j,1}\fita_{j,1} & \gamma_{j-1}\fita_{j,1}\fita_{j,2} & \gamma_{j-1}\fita_{j,1}\fita_{j,3} & \gamma_{j-1}\fita_{j,1}\fita_{j,4} \\
\gamma_{j}\fita_{j,2}\fita_{j,1} & \gamma_{j}\fita_{j,2}\fita_{j,2} & \gamma_{j}\fita_{j,2}\fita_{j,3} & \gamma_{j}\fita_{j,2}\fita_{j,4} \\
\gamma_{j}\fita_{j,3}\fita_{j,1} & \gamma_{j}\fita_{j,3}\fita_{j,2} & \gamma_{j}\fita_{j,3}\fita_{j,3} & \gamma_{j}\fita_{j,3}\fita_{j,4} \\
\gamma_{j+1}\fita_{j,4}\fita_{j,1} & \gamma_{j+1}\fita_{j,4}\fita_{j,2} & \gamma_{j+1}\fita_{j,4}\fita_{j,3} & \gamma_{j+1}\fita_{j,4}\fita_{j,4} \\
\end{pmatrix}\\
\fita_j' {\fita_j'}^* \Gamma_j
&=
\begin{pmatrix}
\gamma_{j-1}\fita'_{j,1}\fita'_{j,1} & \gamma_{j}\fita'_{j,1}\fita'_{j,2} & \gamma_j\fita'_{j,1}\fita'_{j,3} & \gamma_{j+1}\fita'_{j,1}\fita'_{j,4} \\
\gamma_{j-1}\fita'_{j,2}\fita'_{j,1} & \gamma_{j}\fita'_{j,2}\fita'_{j,2} & \gamma_j\fita'_{j,2}\fita'_{j,3} & \gamma_{j+1}\fita'_{j,2}\fita'_{j,4} \\
\gamma_{j-1}\fita'_{j,3}\fita'_{j,1} & \gamma_{j}\fita'_{j,3}\fita'_{j,2} & \gamma_j\fita'_{j,3}\fita'_{j,3} & \gamma_{j+1}\fita'_{j,3}\fita'_{j,4} \\
\gamma_{j-1}\fita'_{j,4}\fita'_{j,1} & \gamma_{j}\fita'_{j,4}\fita'_{j,2} & \gamma_j\fita'_{j,4}\fita'_{j,3} & \gamma_{j+1}\fita'_{j,4}\fita'_{j,4} \\
\end{pmatrix}
\end{align*}
We estimate the difference that comes from changing the index of $\gamma$ in the first matrix to align it with the second matrix.
After that, we bound the difference using the inequality $\|A\|\le \|A\|_1$, where $\|A\|_1$ is the sum of the absolute values of all the entries of $A$.
Using that all the involved number are non-negative, all $\gamma_j\le 1$, and the estimates from \rf{lem:fita}, we obtain:
\begin{align*}
\|\Gamma_j{ \fita_j} \fita_j^* - \fita_j' {\fita_j'}^* \Gamma_j\|
&\le
\frac4t \sB[(\fita_{j,1}+\fita_{j,4})(\fita_{j,2}+\fita_{j,3})+\fita_{j,1}\fita_{j,4}] 
+ \norm|\fita_j' \fita_j'^* - \fita_j \fita_j^* |_1\\
& = O\sC[\frac 1t]\sC[\sqrt{\frac kn} + \sqrt{\frac jk}] + \norm|\fita_j' \fita_j'^* - \fita_j \fita_j^* |_1.
\end{align*}
For the latter, using \rf{lem:fita} again:	
\begin{align*}
\norm|\fita_j' \fita_j'^* - \fita_j \fita_j^* |_1
&\le 
\norm|\fita_j' \fita_j'^* - \fita_j \fita_j'^* |_1
+
\norm|\fita_j \fita_j'^* - \fita_j \fita_j^* |_1
\\
&\le
\s[\norm|\fita'_j|_1 + \norm|\fita_j|_1]  \norm|\fita'_j - \fita_j |_1 
= O\sC[\eps \sqrt{\frac kn} + \eps \sqrt{\frac jk} ].
\end{align*}
Combining the last two equations:
\[
\|\Gamma_j{ \fita_j} \fita_j^* - \fita_j' {\fita_j'}^* \Gamma_j\| = 
O\sC[\frac 1t + \eps]\sC[\sqrt{\frac kn} + \sqrt{\frac jk}].
\]
By \rf{lem:DeltaReflNorm}, and using that $ \Gamma_j{ \fita_j} \fita_j^* - \fita_j' {\fita_j'}^* \Gamma_j = 0 $ for $j>t+1$ due to~\rf{eqn:gammaj}, we get the required bound.
\pfend

\begin{prp}
We have $\displaystyle \|\Gamma\circ \Delta_i \| = O\sC[\frac 1t + \eps]\sqrt{\frac kn}$.
\end{prp}

\pfstart
We use \rf{lem:DeltaINorm}.  
We estimate the first difference 
\[
\bigg|\frac{\sqrt{(k-j)(n- k'-j)}}{n-2j} \gamma_j - \frac{\sqrt{(k'-j)(n- k-j)}}{n-2j} \gamma_{j+1}\bigg|
\]
in the formulation of the lemma, the second one being similar.
It is bounded by
\begin{align*}
&O\sC[\sqrt{\frac kn}] \absA|\gamma_{j+1} - \gamma_j|
+ O\sC[\frac1{\sqrt n}] \absB|\sqrt{k-j} - \sqrt{k'-j}|
+ O\sC[\frac{\sqrt k}{n}] \absB|\sqrt{n-k'-j}-\sqrt{n-k-j}| \\
&\qquad = O\sC[\sqrt{\frac kn}]\frac 1t  +
O\sC[\frac1{\sqrt n}] \eps\sqrt k + 
O\sC[\frac{\sqrt k}{n}] \frac{\eps k}{\sqrt n}.\qedhere
\end{align*}
\pfend

\subsection{Proof of \rf{thm:main}}
\label{sec:proof}
Let us gather up all the inequalities from \rf{sec:construction}.
We assume that
\begin{equation}
\label{eqn:ell}
2\ell\le t \le k/5,\qquad \ell \le C \min\sfigC{\frac{\sqrt k}{\eps},\;  \frac{n}{k\eps^2}},
\end{equation}
and we have that
\begin{align}
\norm|\Gamma\circ\Dp|,\, \norm|\Gamma\circ\Dps|
&=O\sC[ \eps\sqrt{\frac kn} + \eps \sqrt{\frac tk} + \frac 1t ],
\label{eqn:nu}\\
\|\Gamma\circ\Dpp\| 
&= O\sC[\frac 1t + \eps]\sC[\sqrt{\frac kn} + \sqrt{\frac tk}], 
\label{eqn:nunu*}\\
\|\Gamma\circ \Delta_i \| &= O\sC[\frac 1t + \eps]\sqrt{\frac kn}.
\label{eqn:i}
\end{align}


The proof of the Theorem follows from the two lemmata below.
In both of them, we assume notation of \rf{thm:main}.

\begin{lem}
\label{lem:case1}
Assume the algorithm has $\ell$ copies of the state $\psi_x$ where $\ell$ satisfies the conditions of~\rf{eqn:ell}, the first inequality reading as $2\ell\le k/5$.
Assume also that the algorithm uses the state-generating oracle
\begin{equation}
\label{eqn:ell'}
\QG \le C \min\sfigC{\frac1\eps \sqrt{\frac nk},\;\; \frac1\eps \sqrt{\frac k\ell}, \;\; \frac{k^{1/3}}{\eps^{2/3}} }
\end{equation}
times, where $C$ is a sufficiently small constant.
Then, in order to solve the problem, the algorithm has to
\begin{itemize}
\item either
execute the reflecting oracle
\( \displaystyle
\Omega\sD[\min\sfigC{\frac1\eps \sqrt{\frac nk},\;\; \frac1\eps \sqrt{\frac k{\ell+\QG}},\;\; \sqrt{\frac k\eps} }]
\)
times;

\item or
execute the membership oracle 
$\displaystyle \Omega \s[\frac1\eps\sqrt{\frac nk}]$ 
times;
\end{itemize}
\end{lem}

\begin{lem}
\label{lem:case2}
Assume the algorithm does not have any copies of the state and does not execute the state-generating oracle.  
Then, in order to solve the problem, it has to
\begin{itemize}
\item execute the reflecting or the membership oracle 
\( \displaystyle
\Omega\sC[\sqrt{\frac nk}]
\)
times.
\end{itemize}
\end{lem}

\begin{proof}[Proof of \rf{thm:main} assuming Lemmata~\ref{lem:case1} and~\ref{lem:case2}]
\rf{thm:main} states that, in order to solve the problem, the algorithm must satisfy one of the conditions in Rows 1 through 8 of \rf{tbl:main}.
Let us prove that if Rows 1 through 7 do not hold, then Row 8 holds.

From the falsity of Row 1 and the falsity of Rows 3 and 4, we get that, respectively, the conditions of~\rf{eqn:ell} and~\rf{eqn:ell'} are satisfied.
These being the two conditions of \rf{lem:case1}, one of the two bullets of \rf{lem:case1} must hold. 
However, the falsity of Row 2 rules out the second bullet, so the first bullet must hold.

Further on, from the falsity of Rows 5 and 6, we have that
\[q_R = o\s[\min\sfig{
   \frac1\eps \sqrt{\frac nk},\;
   \frac 1\eps \sqrt{\frac{k}{\ell+q_G}},\;
   \sqrt{\frac k\eps} + \sqrt{\frac nk}
}].\]
Contrasting this with the truth of the first bullet of \rf{lem:case1}, we see that both $q_R = \Omega(\sqrt{k/\eps})$ and $q_R = o(\sqrt{n/k})$.
The former, together with the falsity of Row 7, implies that $\ell = q_G = 0$, therefore the conditions of \rf{lem:case2} are satisfied. 
By this lemma, since $q_R = o(\sqrt{n/k})$, it must be that $q_M = \Omega(\sqrt{n/k})$, and thus Row 8 holds.
\end{proof}

%
%

\begin{proof}[Proof of \rf{lem:case1}]
In the proof, we will need a bit more careful tracking of constants than usually.
We will use two constants: $C$ as in the statement of the lemma, and one additional constant $C'$.  The $O$s and $\Omega$s in the proof of the lemma do not depend on $C$ and $C'$, whereas the $\Omega$s in the formulation of the lemma do depend on them.

We will take
\[
t = \max\sfig{ 2\ell,\;\; C'\QG,\;\; \frac1{5\eps} },
\]
where $C'$ is a sufficiently large constant.
First, we have to check that $t$ satisfies~\rf{eqn:ell}.
Indeed, $t\ge 2\ell$.  Also, the three possible values of $t$ are upper-bounded as
\[
2\ell\le k/5,
\qquad
C'\QG \le C'C \frac{k^{1/3}}{\eps^{2/3}} \le \frac k5,
\qqand
\frac1{5\eps}\le \frac{k}{5},
\]
if $C$ is small enough and because $\eps\ge 1/k$.

Now we can use~\rf{eqn:nu}--\rf{eqn:i}.
Since $t = \Omega(1/\eps)$, we have that $O\sA[\eps+1/t] = O(\eps)$.
\rf{prp:formulation} gives us that at least one of the following three cases holds:
\begin{itemize}
\item The state-generating oracle is used
\begin{equation}
\label{eqn:ellEstimate}
\Omega\sD[\min\sfigC{ 
    t,\;\; \frac1\eps \sqrt{\frac nk},\;\; \frac1\eps \sqrt{\frac kt}
}]
=
\Omega\sD[\min\sfigC{ 
    t,\;\; \frac1\eps \sqrt{\frac nk},\;\; \\
    \frac1\eps \sqrt{\frac k\ell},\;\; \frac1\eps \sqrt{\frac k{C'\QG}},\;\; \sqrt{\frac k\eps}
}]
\end{equation}
times; or
\item the reflecting oracle is used
\[
\Omega\sD[\min\sfigC{ 
    \frac1\eps \sqrt{\frac nk},\;\; \frac1\eps \sqrt{\frac kt}
}]
=
\Omega\sD[\min\sfigC{ 
    \frac1\eps \sqrt{\frac nk},\;\; \frac1\eps \sqrt{\frac k{\ell+C'\QG}},\;\; \sqrt{\frac k\eps}
}]
\]
times; or
\item the membership oracle is used
\(
\displaystyle 
\Omega\sC[\frac1\eps \sqrt{\frac nk}]
\)
times.
\end{itemize}

Now, in order to prove the lemma, it suffices to show that the first case cannot hold, that is, $\QG$ is smaller than~\rf{eqn:ellEstimate}.
For that, we have to compare $\QG$ to the five elements in the minimum on the right-hand side of~\rf{eqn:ellEstimate}.

First, $t\ge C'\QG$.  This gives contradiction to $\QG = \Omega(t)$ if $C'$ is large enough.
Second and third, we have $\QG \le C \frac1\eps \sqrt{\frac n k}$ and $\QG \le C \frac 1\eps \sqrt{\frac k\ell}$, and we can take $C$ small enough.
Forth, we have
\[
\QG \le C\frac{k^{1/3}}{\eps^{2/3}}
\quad\Longrightarrow\quad
\QG^{ 3/2} \le C^{3/2} \frac{\sqrt k}{\eps}
\quad\Longrightarrow\quad
\QG \le C^{3/2} \frac1\eps \sqrt{\frac k{\QG}}
\]
and we can take $C$ small enough.
Fifth and finally, we have
\[
\QG \le C\frac{k^{1/3}}{\eps^{2/3}}
\le C \sqrt{\frac k\eps}
\]
because $\eps \ge 1/k$.
Again, we can take $C$ small enough.

The order of choosing the values of the constants above is as follows.
First, we choose the value of $C'$ large enough to get contradiction to $\QG = \Omega(t)$.
Then, based on the value of $C'$, we choose the value of $C$ small enough.
This ends the proof of \rf{lem:case1}
\end{proof}

\begin{proof}[Proof of \rf{lem:case2}]
In this case we have $\ell = \QG = 0$.
We take $t=1$.  Thus, $\gamma_0=1$ and $\gamma_j = 0$ for all $j\ge 1$.
Clearly~\rf{eqn:ell} is satisfied. 

By~\rf{eqn:i}, we have that%
\footnote{
This estimate can be also obtained directly, since in this case $\Gamma$ is the normalised all-1 matrix.
}
\[
\frac 1{\norm|\Gamma\circ\Delta_i|} = \Omega \s[\sqrt{\frac nk}].
\]
We will use a different analysis of $\|\Gamma\circ\Dpp \|$, tailored for this special case. Note that 
\(
\fita_{0,1}=\fita'_{0,1} = \fita_{0,3} = \fita'_{0,3} = 0,
\)
hence, we get that the matrix 
$\Gamma_0{ \fita_0} \fita_0^* - \fita_0' \fita_0' \Gamma_0$
from~\rf{lem:DeltaReflNorm} equals
\[
\begin{pmatrix}
0 & 0 & 0 & 0 \\
0 & \fita^{2}_{0,2} -\fita_{0,2}^{\prime\; 2} & 0 & \fita_{0,2}\fita_{0,4} \\
0 & 0 & 0 & 0 \\
0 & -\fita'_{0,2}\fita'_{0,4} & 0 & 0 \\
\end{pmatrix},
\]
and
\[
\Gamma_1{ \fita_1} \fita_1^* - \fita_1' \fita_1'^*\Gamma_1= 
\begin{pmatrix}
\fita^{2}_{1,1} - \fita^{\prime\; 2}_{1,1} &  \fita_{1,1}\fita_{1,2} &  \fita_{1,1}\fita_{1,3} &  \fita_{1,1}\fita_{1,4} \\
-\fita'_{1,1}\fita'_{1,2} & 0 & 0 & 0 \\
-\fita'_{1,1}\fita'_{1,3} & 0 & 0 & 0 \\
-\fita'_{1,1}\fita'_{1,4} & 0 & 0 & 0 \\
\end{pmatrix}
\]
and all the remaining matrices are zeroes.
Now we can use \rf{lem:DeltaReflNorm}.  From \rf{lem:fita}, it is straightforward to see that
\[
\frac 1{\norm|\Gamma\circ\Dpp|} =
\Omega \s[\sqrt{\frac nk}].
\]
Application of \rf{prp:formulation} finishes the proof of \rf{lem:case2}.
\end{proof}

\newcommand{\topp}{*}

\renewcommand{\inn}{1}
\renewcommand{\out}{0}

\renewcommand{\inn}{{\in}}
\renewcommand{\out}{{\notin}}

\renewcommand{\inn}{i}
\renewcommand{\out}{o}

\renewcommand{\inn}{\mathrm{i}}
\renewcommand{\out}{\mathrm{o}}

\newcommand{\beEs}{\begin{equation*}}
\newcommand{\enEs}{\end{equation*}}
\def\beA#1\enAs{\begin{align*}#1\end{align*}}
\def\beM#1\enMs{\begin{multline*}#1\end{multline*}}

\section{Preliminaries on Representation Theory}
\label{sec:JohnsonPrelim}

To conclude the proof of our lower bounds, we are left with two tasks. The first one is to define the morphisms $\Phi_j$ used in the construction of the adversary matrix $\Gamma$. 
The second one is to prove Lemmata~\ref{lem:CoeffsUnderA}--\ref{lem:DeltaINorm}. 
In this section, we introduce the basics of representation theory of the symmetric group, and use them to define morphisms $\Phi_j$. 
We leave the proofs of Lemmata~\ref{lem:CoeffsUnderA}--\ref{lem:DeltaINorm} to the final three sections.

\subsection{Representation theory}
\label{sec:representation}

In this section, we introduce basic notions from representation theory of finite groups.  For more background, the reader may refer to~\cite{curtis:representationTheory, serre:representation}.

Let $G$ be a finite group.
The \emph{group algebra} $\bC G$ consists of formal linear combinations of the form $\sum_{g\in G}\alpha_g g$ with $\alpha_g\in\bC$.
It is clearly a linear space, but it is also an algebra because we can extend the multiplication law of $G$ to all $\bC G$ by linearity.
%
A \emph{left module} over $\bC G$ is a vector space $\cM$ with a left multiplication operation by the elements of $\bC G$ satisfying the usual associativity and distributivity axioms and such that $ev=v$ for $e$ the identity in $G$ and every $v\in\cM$.
Such modules are also known as representations of $G$.
We will use the term \emph{$G$-module}, which is a standard~\cite{sagan:symmetricGroup, james:symmetricGroup} misnomer, since what is actually meant is a $\bC G$-module. 
We can treat elements of $\bC G$ as linear operators acting on $\cM$.
Due to linearity, it suffices to specify these operators for $g\in G$. 

We assume the module $\cM$ is equipped with a \emph{$G$-invariant} inner product, that is, $\ip<gv,gu> = \ip<u,v>$ for all $u,v\in \cM$ and $g\in G$.
Thus, the linear operators corresponding to $g\in G$ are unitary.
If $\cV$ and $\cW$ are two $G$-modules, then the direct sum $\cV\oplus \cW$ and the tensor product $\cV\otimes \cW$ are also $G$-modules defined by $g(v,w) = (gv, gw)$ and $g(v\otimes w) = (gv)\otimes (gw)$ for all $v\in\cV$, $w\in\cW$, and $g\in G$.

\mycutecommand{\Hom}{\mathop{\mathrm{Hom}}\nolimits}
\mycutecommand{\Iso}{\mathop{\mathrm{Iso}}\nolimits}
\mycutecommand{\Slice}{\mathop{\mathrm{Slice}}\nolimits}

A \emph{$G$-morphism} (or just morphism, if $G$ is clear from the context) between two $G$-modules $\cV$ and $\cW$ is a linear operator $\theta\colon \cV\to \cW$ that commutes with all $\alpha\in\bC G$: $\theta\alpha = \alpha\theta$. 
By linearity, it suffices to check commutativity with all $\alpha\in G$. 
Let us denote by $\Hom_G(\cV, \cW)$ the linear space of $G$-morphisms from $\cV$ to $\cW$.

An important special case of a $G$-module is $\bC^X$, where $X$ is a finite set with a \emph{group action of $G$} on it.
A group action is a map $(g,x)\mapsto g(x)$ from $G\times X$ onto $X$ satisfying $g(h(x))=(gh)(x)$ for all $g,h\in G$ and $x\in X$.
By linearity, this gives a $G$-module.
The linear operators on $\bC^X$ corresponding to $g\in G$ are given by permutation matrices. 
Hence, the standard inner product in $\bC^X$ is $G$-invariant.
We have the following easy characterisation of morphisms in this case:

\begin{prp}
\label{prp:groupAction}
Assume $X$ and $Y$ are two sets with group action of $G$ defined on them.
A linear operator $A\colon \bC^Y\to \bC^X$ is a $G$-morphism if and only if $A\elem[x,y] = A\elem[g(x), g(y)]$ for all $x\in X$, $y\in Y$, and $g\in G$.
\end{prp}

A $G$-module is called \emph{irreducible} (or irrep) 
if it does not contain a non-trivial $G$-submodule.
Schur's Lemma is an essential result in representation theory, stated as follows.
\begin{lem}[Schur's Lemma]
\label{lem:schur}
Assume $\theta\colon \cV\to \cW$ is a morphism between two irreducible $G$-modules $\cV$ and $\cW$. 
If $\cV$ and $\cW$ are non-isomorphic, then  $\theta=0$.  
Otherwise, $\theta$ is uniquely determined up to a scalar multiplier.
\end{lem}

In other words, the second part of the above lemma states that $\Hom_G(\cV, \cV)\cong \bC$ for irreducible $\cV$.
Let us note that while the first part holds for any base field, for the second part it is essential that the base field $\bC$ is algebraically closed.

Schur's Lemma has a large number of consequences.
First, on an irrep $\cV$, there is only one, up to a scalar multiplier, $G$-invariant inner product.
Hence, any isomorphism between two irreducible $G$-modules is an isometry times a scalar.
Second, if $\cM$ is a $G$-module, and $\cV$ and $\cW$ are two its non-isomorphic irreducible submodules, then $\cV$ and $\cW$ are orthogonal as subspaces.

Let $\cV$ be an irrep.
A submodule of $\cM$ isomorphic to $\cV$ is called a \emph{copy} of $\cV$ in $\cM$.
The dimension of $\Hom_G(\cV,\cM)$ is known as the \emph{multiplicity} of $\cV$ in $\cM$.
It is equal to the number of pairwise orthogonal copies of $\cV$ that can be embedded in $\cM$.

Next, let $\Iso_G(\cV, \cM)$ denote the span of all the copies of $\cV$ in $\cM$, which is known as the \emph{isotypical subspace}.
The mapping $\theta\otimes v \mapsto \theta(v)$ defines an isomorphism between $\Hom_G(\cV, \cM)\otimes\cV$ and $\Iso_G(\cV, \cM)$ as linear spaces.
The whole module $\cM$ can be expressed as $\cM = \bigoplus_{\cV}\Iso_G(\cV, \cM)$ as $\cV$ ranges over the irreps of $G$.
This is called the isotypical decomposition of $\cM$.
The terms in the isotypical decomposition are pairwise orthogonal.

We use the following palpable way to represent $\Hom_G(\cV,\cM)$.
Fix some reference non-zero vector $v_\cV\in\cV$, and define the \emph{slice} of $v_\cV$ in a module $\cM$, denoted $\Slice_G(v_\cV, \cM)$, as the set of all $\theta(v_\cV)$ as $\theta$ ranges over $\Hom_G(\cV,\cM)$.
The space $\Hom_G(\cV,\cM)$ is isomorphic to $\Slice_G(v_\cV, \cM)$, the isomorphism being $\theta\mapsto \theta(v_\cV)$.
Combining the two above isomorphisms, we get that $\Iso_G(\cV, \cM)$ is isomorphic to $\Slice_G(v_\cV,\cM)\otimes \cV$.
This allows us to write down the following isotypical decomposition  of a morphism $\theta\colon\cM\to\cN$ between two $G$-modules $\cM$ and $\cN$:
\begin{equation}
\label{eqn:isotypical}
\theta = \bigoplus_\cV \theta_\cV\otimes I_\cV.
\end{equation}
Here $\cV$ ranges over all irreps of $G$,
a linear operator $\theta_\cV\colon \Slice_G(v_\cV, \cM)\to\Slice_G(v_\cV, \cN)$ is the restriction of $\theta$ to $\Slice_G(v_\cV, \cM)$ for a fixed non-zero $v_\cV\in \cV$, and $I_\cV$ is the identity on $\cV$.
In particular,~\rf{eqn:isotypical} implies 
\begin{equation}
\label{eqn:isotypicalNorm}
\|\theta\| = \max_\cV \|\theta_\cV\|.
\end{equation}

We will use group algebra extensively in our proofs.
Its elements are important as they give linear operators that are preserved under isomorphisms.
In particular, if $\alpha\in \bC G$ is a group element, and $\alpha_\cM$ is the corresponding linear operator in a $G$-module $\cM$, we have the following decomposition complementing~\rf{eqn:isotypical}:
\begin{equation}
\label{eqn:GroupElement}
\alpha_\cM = \bigoplus_\cV I_{\Hom_G(\cV,\cM)}\otimes \alpha_\cV,
\end{equation}
where again $\cV$ ranges over all irreps of $G$.
An important use-case is provided by an $\alpha$ that is zero on all irreps but one (call it $\cV$) and is a 1-dimensional projector on $\cV$.
In this case, by~\rf{eqn:GroupElement}, $\alpha$ projects onto the slice of $v$ in every $G$-module $\cM$, where $v$ is a non-zero vector in $\alpha(\cV)$.

\mycutecommand{\grS}{\EuScript{S}}

\subsection{Symmetric group}
In this section, we consider a special case of the symmetric group.
If $A$ is a finite set, $S_A$ denotes the \emph{symmetric group} on $A$, that is, the group with the permutations of $A$ as elements, and composition as the group operation.
We will write $S_n$ instead of $S_{[n]}$.
If $A\subseteq B$, then we consider $S_A$ as a subgroup of $S_B$, where each permutation $\pi\in S_A$ is extended to $B\setminus A$ by identity.
We also use the following elements of the group algebra $\bC S_A$:
\[
\grS_A^+ = \sum_{\pi\in S_A} \pi
\qqand
\grS_A^- = \sum_{\rho\in S_A} (\sgn\rho)\rho,
\]
where $\sgn\rho$ stands for the sign of $\rho$.
Both of them are scalar multiples of orthogonal projectors.

Representation theory of $S_n$ is closely related to \emph{partitions of integers}.
A \emph{partition} $\lambda$ of an integer $n$ is a non-increasing sequence $(\lambda_1,\dots,\lambda_\ell)$ of positive integers satisfying $\lambda_1+\dots+\lambda_\ell = n$.  
For each partition $\lambda$ of $n$, one assigns an irreducible $S_n$-module $\Spe{\lambda}$, called the \emph{Specht module}.  All these modules are pairwise non-isomorphic, and give a complete list of all the irreps of $S_n$.

\mycutecommand{\grC}{\EuScript{C}}
\mycutecommand{\grR}{\EuScript{R}}

A partition $\lambda = (\lambda_1,\dots,\lambda_\ell)$ of $n$ is often represented in the form of a \emph{Young diagram} that consists, from top to bottom, of rows of $\lambda_1,\lambda_2,\dots,\lambda_\ell$ boxes aligned by the left side.
We often identify $\lambda$ with the corresponding diagram.
A \emph{Young tableau} of shape $\lambda$ is a Young diagram of $\lambda$ with each box containing an element from $[n]$, each element used exactly once.
As an example, the following is a Young tableau of shape $(4,3,1)$:
\[
t = \begin{ytableau}
4&6&5&1 \\ 2&3&8\\7
\end{ytableau}
\]

For a Young tableau $t$, define its \emph{row permutations} $R_t$ and its \emph{column permutations} $C_t$ as the permutations in $S_n$ that permute the elements within each row or column of $t$, respectively.
In our example above, $R_t = S_{\{1,5,6,4\}}\times S_{\{2,3,8\}}$ and $C_t = S_{\{2,4,7\}}\times S_{\{3,6\}}\times S_{\{5,8\}}$.
This gives rise to the following two elements of the group algebra $\bC S_n$:
\[
\grR^+_{t} = \sum_{\pi\in R_{t}}\pi 
= \prod_{\text{$R$ is a row of $t$}} \grS^+_R ,
\qqand
\grC^-_{t} = \sum_{\rho\in C_{t}}(\sgn\rho)\rho
= \prod_{\text{$C$ is a column of $t$}} \grS^-_C.
\]

The following is one of the key results in representation theory of the symmetric group.
In fact, this is the only result from this theory we rely on in our forthcoming proofs.

\begin{thm}[{\cite[Chapter 3]{james:symmetricGroup}}, {\cite[\S 28]{curtis:representationTheory}}]
\label{thm:Et}
Let $t$ be a Young tableau of shape $\lambda$.
The element $\grC^-_{t}\grR^+_t$ of the group algebra annihilates (is identical zero on) every Specht module $\cS^\mu$ with $\mu\ne\lambda$.
On $\cS^\lambda$, it is a 1-dimensional projector times a scalar.
\end{thm}

This result is important as $\grC^-_t\grR^+_t$ is the projector onto $\Slice_{S_n}(v, \cM)$ in every $S_n$-module $\cM$, where $v$ lies in the image of $\grC^-_{t}\grR^+_t$ in $\cS^\lambda$.
Note that while both $\grC^-_{t}$ and $\grR^+_t$ are non-normalized orthogonal projectors, they generally do not commute, hence, $\grC^-_{t}\grR^+_t$ is not Hermitian, and is \emph{not} an \emph{orthogonal} projector.

\newcommand{\Johnson}[2]{\bC^{\binom{[#1]}{#2}}}
\mycutecommand{\GrProjector}{\grC^-_{t_{j}} \grR^+_{t_j}}

\subsection{The Module \texorpdfstring{$\Johnson nk$}{C([n] choose k)}} \label{sec:AssocPrelim}

For the remaining part of this paper, we assume the value of $n$ is fixed, so our notation will often implicitly depend on $n$.

For positive integers $n\ge 2k$, the set $\binom{[n]}{k}$ of all subsets of $[n]$ of size $k$ admits group action of $S_n$ on it.
Namely, $x\in\binom{[n]}{k}$ is mapped by $\pi\in S_n$ to $\pi(x)=\sfig{\pi(i)\mid i\in x}$.
Let $\Johnson nk$ be the corresponding $S_n$-module.
We will write the elements of $\Johnson nk$ as formal linear combinations of subsets similar as we write group algebra elements.
This module has close connection to the Johnson association scheme~\cite{bannai:algebraicCombinatorics}, and both spaces $\bC^X$ and $\bC^Y$ in \rf{sec:lower} are of this form.

In this section, we will use \rf{thm:Et} to analyse its structure.
Although the results are well-known, this serves as a warm-up for the next sections.
We will need the following piece of notation.
For two disjoint sets $A, B\subseteq [n]$, let $A\sqtimes B$ denote their disjoint union.
We extend this notation to formal linear combinations of sets in a straightforward way. 
For instance,
\[
\sB[\{1\} - \{2\}]\sqtimes \sB[\{3,5\}-\{4,5\}] = \{1,3,5\}-\{1,4,5\}-\{2,3,5\}+\{2,4,5\}.
\]
This notation is similar to tensor products with the distinction that the sets are unordered and the promise that the multipliers are disjoint.
The symbol $\sqtimes$ is a combination of the symbols for the disjoint union $\sqcup$ and the product $\times$.

Take any 3 distinct elements $a,b,c\in [n]$ and $T\in {\binom{[n]}k}$.
Either $\absA|\{a,b,c\}\cap T|\ge 2$, or $\absA|\{a,b,c\}\cap ([n]\setminus T)|\ge 2$.
In either case, $\grS^-_{\{a,b,c\}}T=0$, hence, $\grS^-_{\{a,b,c\}}$ annihilates the whole $\Johnson nk$.
By \rf{thm:Et}, $\Johnson nk$ does not contain a copy of any $S^\lambda$ for $\lambda$ with more than two rows.

Now consider two-row diagrams $\lambda_j = (n-j,j)$, which also includes the one-row case of $\lambda_0 = (n)$.
For notational convenience, we choose the following tableau of shape $\lambda_j$:
\begin{equation}
\label{eqn:t_j}
\ytableausetup{mathmode, boxsize=2em}
t_j = \begin{ytableau}
a_1 & a_2 & a_3 & \none[\dots] & a_j & 1 & 2 & 3 & \none[\dots] & \scriptstyle n-2j \\
b_1 & b_2 & b_3 & \none[\dots] & b_j
\end{ytableau}\;\;,
\ytableausetup{boxsize=normal}
\end{equation}
where $a_i = n-2i+2$ and $b_i = n-2i+1$ for $i=1,\dots,j$.
If $T\subseteq [n]$, then $\grC^-_{t_j}T=0$ unless it satisfies the following intersection condition:
\begin{equation}
\label{eqn:intersectionCondition}
\forall i\in [j]\colon \absA|T\cap \{a_i,b_i\}|=1.
\end{equation}
By \rf{thm:Et}, there is no irrep isomorphic to $\Spe{(n-j,j)}$ in $\Johnson nk$ for $k<j$.

Now let us consider the case $k=j$.
If $T\in {\binom{[n]}j}$ and satisfies the intersection condition~\rf{eqn:intersectionCondition}, then $\grC^-_{t_j}T = (-1)^{|T\cap\{b_1,\dots,b_j\}|}C_j$, where
\begin{equation}
\label{eqn:C_j}
C_j = \sB[\{a_1\}-\{b_1\}]\sqtimes \sB[\{a_2\}-\{b_2\}]\sqtimes\cdots\sqtimes \sB[\{a_j\}-\{b_j\}].
\end{equation}

For any $\pi\in R_{t_j}$, the sets $T\cap\{b_1,\dots,b_j\}$ and $\pi(T)\cap\{b_1,\dots,b_j\}$ have the same size, hence, the image of $\grC^-_{t_j}\grR^+_{t_j}$ is non-empty on $\Johnson nj$, as there are no cancellations.
The image is spanned by $C_j$, which, by \rf{thm:Et}, means that $C_j$ belongs to the only copy of $\Spe{(n-j,j)}$ contained in $\Johnson nj$.
We will use the latter as our reference instance of $\Spe{(n-j,j)}$, and $C_j$ as our reference vector in $\Spe{(n-j, j)}$.
Note that, while they are related, $C_j$ is an element of the module $\Johnson nj$, and $\grC^-_{t_j}$ is an element of the group algebra $\bC S_n$.

For the general case $k\ge j$, we can use the same logic as above with $\grC_{t_j}^-$ replaced by $\grC^-_{t_j}\grS^+_{[n-2j]}$.
As above, for any $T\in {\binom{[n]}k}$, the vector $\grC^-_{t_j}\grS^+_{[n-2j]} T$ is a scalar multiple of the vector
\begin{equation}
\label{eqn:R^A_ell}
C_j\sqtimes R^{[n-2j]}_{k-j}
\qquad\text{with}\qquad
R^A_s = \sum_{B\subset A, |B|=s} B,
\end{equation}
where the scalar is 0 unless the intersection condition~\rf{eqn:intersectionCondition} is met.
If the latter is satisfied, the scalar is $(-1)^{|T\cap\{b_1,\dots,b_j\}|} (k-j)! (n-j-k)!$.
We have that $S_{[n-2j]}$ is a subgroup of  $R_{t_j}$, thus, $\grS^+_{[n-2j]} \grR^+_{t_j} = (n-2j)!\, \grR^+_{t_j}$.
Using the same reasoning as above, there are no cancellations again, and the image $\grC^{-}_{t_j} \grR^+_{t_j}\s[\Johnson nk] = \grC^{-}_{t_j} \grS^+_{[n-2j]} \grR^+_{t_j}\s[\Johnson nk] $ is one-dimensional and spanned by $C_j\sqtimes R^{[n-2j]}_{k-j}$ of~\rf{eqn:R^A_ell}.
In light of \rf{thm:Et}, we have proven the following result:

\begin{thm}
\label{thm:Johnson}
We have the following decomposition into irreps, each irrep having multiplicity 1:
\begin{equation}
\label{eqn:JohnsonIrreps}
\Johnson nk \cong \Spe{(n)} \oplus \Spe{(n-1,1)} \oplus \Spe{(n-2,2)} \oplus\ldots\oplus \Spe{(n-k,k)}.
\end{equation}
Also, for $j\le k$, $\Slice_{S_n}\sB[C_j, \Johnson nk]$ is one-dimensional and spanned by the following normalised vector 
\begin{equation}
\label{eqn:v_jk}
v_j^k = \frac{1}{\sqrt{2^j{\binom{n-2j}{k-j}}}}\; C_j\sqtimes R^{[n-2j]}_{k-j}.
\end{equation}
On $\Johnson nk$, both $\grC^-_{t_j} \grR^+_{t_j}$ and $\grC^-_{t_j}\grS^+_{[n-2j]}$ are a scalar times a projector onto $v_j^k$.
\end{thm}

We will need the following small technical result.
\begin{clm}
\label{clm:Sabc}
Assume $a,b,c\in [n]$ are 3 distinct elements.
Then, $\grS^-_{\{a,b,c\}}$ annihilates any irrep of the form $S^{(n-j,j)}$.
\end{clm}

\begin{proof}
As we saw earlier, $\grS^-_{\{a,b,c\}}$ annihilates the whole $\Johnson nk$.
Hence, it annihilates all the terms in~\rf{eqn:JohnsonIrreps}, which proves the claim.
\end{proof}

We can also describe the isometric isomorphisms $\Phi^{\ell \to k}_j$ between the copies of $S^{(n-j,j)}$ in $\Johnson n\ell$ and $\Johnson nk$.
They transform the corresponding normalised vectors~\rf{eqn:v_jk}:
\begin{equation}
\label{eqn:Phil->k}
\Phi^{\ell\to k}_j \colon 
v_j^\ell \mapsto v_j^k.
\end{equation}
Since it is a morphism, it also transforms 
$
\Phi^{\ell\to k}_j \colon 
\pi\sA[v_j^\ell] \mapsto \pi\sA[v_j^k]
$
for every choice of $\pi\in S_n$.

For the ease of notation, we will prove Lemmata~\ref{lem:CoeffsUnderA}--\ref{lem:DeltaINorm} for the operator $\Gamma\colon \Johnson n\ell \to \Johnson nk$ given by 
\begin{equation}
\label{eqn:GammaWithell}
\Gamma = \bigoplus_j \gamma_j \Phi_j^{\ell\to k}.
\end{equation}
The statements of the Lemmata are then obtained substituting $k'$ instead of $\ell$.

We will also need a more convenient way of going from $\Johnson n\ell$ to $\Johnson nk$, where $k > \ell$.  Define the following map
\begin{equation}
\label{eqn:W_ell^k}
W^{\ell\to k}\colon \Johnson n\ell \to \Johnson nk,\quad T\mapsto T\sqtimes R^{[n]\setminus T}_{k-\ell} = \sum_{S: T\subseteq S\subseteq [n], |S|=k} S.
\end{equation}
In other words, $W^{\ell\to k}\elem[S,T] = 1_{T\subseteq S}$.  The latter condition is preserved under the action of any $\pi\in S_n$, hence, by \rf{prp:groupAction}, $W^{\ell\to k}$ is a morphism.

\begin{clm}
\label{clm:Waction}
For a subset $A\subseteq [n-2j]$ of size $\ell-j$, we have
\begin{equation}
\label{eqn:Waction}
W^{\ell\to k}\colon\quad C_j\sqtimes A\; \mapsto\;  C_j\sqtimes A \sqtimes R^{[n-2j]\setminus A}_{k-\ell}.
\end{equation}
\end{clm}

\pfstart
Indeed, take any basis element $T$ used on the left-hand side of~\rf{eqn:Waction}.
By construction, $|T\cap \{a_i,b_i\}|=1$ for every $i$.
Assume $W^{\ell\to k}$ extends $T$ to a superset $S\supseteq T$ such that $\{a_i,b_i\}\subseteq S$.
But $S$ can also then be obtained from $T\triangle \{a_i,b_i\}$, where $\triangle$ is the symmetric difference.
The subsets $T$ and $T\triangle\{a_i,b_i\}$ have opposite signs in $C_j\sqtimes A$, hence, $S$ cancels out.

The only terms that are not cancelled out in this fashion come from extensions by the elements in $[n-2j]\setminus A$, which gives~\rf{eqn:Waction}.
\pfend

\begin{prp}
\label{prp:WIsotypical}
The morphism $W^{\ell\to k}$ satisfies
\begin{equation}
\label{eqn:WIsotypicalA}
W^{\ell\to k }\colon\quad C_j\sqtimes R^{[n-2j]}_{\ell-j} \mapsto {\binom{k-j}{\ell-j}} C_j\sqtimes R^{[n-2j]}_{k-j}.
\end{equation}
Therefore, it has the following isotypical decomposition:
\begin{equation}
\label{eqn:WIsotypicalB}
W^{\ell\to k} = \sum_{j=0}^\ell \sqrt{ {\binom{n-j-\ell}{k-\ell}}{\binom{k-j}{\ell-j}}} \Phi^{\ell\to k}_j .
\end{equation}
\end{prp}

\pfstart
Eq.~\rf{eqn:WIsotypicalA} follows from \rf{clm:Waction}, as each $B$ used in $R^{[n-2j]}_{k-j}$ has exactly $\binom{k-j}{\ell-j}$ choices of $A$ in $R^{[n-2j]}_{\ell-j}$ it can be obtained from.
Taking into account the normalisation factors from~\rf{eqn:v_jk}, the coefficient of $\Phi^{\ell\to k}_j$ in the isotypical decomposition of $W^{\ell\to k}$ is
\[
\binom{k-j}{\ell-j} \sqrt{\frac{\binom{n-2j}{k-j}}{\binom{n-2j}{\ell-j}}} = \sqrt{ {\binom{n-j-\ell}{k-\ell}}{k-j\choose \ell-j}}.\qedhere
\]
\pfend

\newcommand{\Bigmodule}[2]{\Johnson{#1}{#2}\otimes \bC^{#1}}
\mycutecommand{\hV}{{V'}}

\section{Proof of \rf{lem:DeltaINorm}}
\label{sec:DeltaINormProof}

Due to symmetry, the norms of all $\Gamma\circ\Delta_i$ are the same, thus, we may consider $\Gamma\circ\Delta_1$.
Recall that we assume that $\Gamma$ in an $\binom{[n]}k\times \binom{[n]}\ell$-matrix from~\rf{eqn:GammaWithell}.
Let $\Pi_{k}^\circ$ and $\Pi_{k}^\bullet$ denote orthogonal projections in $\Johnson nk$ onto the span of 
$\sfigA{T\in \binom{[n]}k \midA 1\notin T}$ 
and 
$\sfigA{T\in \binom{[n]}k \midA 1\in T}$,
respectively.
Then $\Gamma\circ\Delta_1$ decomposes as the following direct sum
\begin{equation}
\label{eqn:GammaDelta1}
\Gamma\circ \Delta_1 = \Pi_{k}^\circ \Gamma \Pi_{\ell}^\bullet \oplus \Pi_{k}^\bullet \Gamma \Pi_{\ell}^\circ,
\end{equation}
and it suffices to estimate the norms of both terms on the right-hand side independently.

This matrix is no longer symmetric with respect to the whole group $S_n$, but it is an $S_{[2..n]}$-morphism, where $[2..n]$ denotes $\{2,3,\dots,n\}$.
The group $S_{[2..n]}$ is clearly isomorphic to $S_{n-1}$, and the vector $C_j$ from~\rf{eqn:C_j} still acts as a reference vector in the irrep $\Spe{(n-1-j,j)}$ of $S_{[2..n]}$.

Let us denote
\[
\cA_j^k = \Slice_{S_{[2..n]}} \sA[C_j, \Johnson nk],
\]
where $\Johnson nk$ is considered as an $S_{[2..n]}$-module.
The proof of the following proposition is essentially identical to \rf{thm:Johnson}.

\renewcommand{\O}{{[\mathrm{O}]}_j}
\newcommand{\I}{{[\mathrm{I}]}_j}

\begin{prp}
If $j< k$, the space $\cA^k_j$ is two-dimensional with the following orthogonal basis:
\[
\O^k = C_j \sqtimes R^{[2..n-2j]}_{k-j}
\qqand
\I^k = C_j \sqtimes \{1\}\sqtimes R^{[2..n-2j]}_{k-j-1},
\]
with the vectors being in the image of $\Pi_{k}^\circ$ and $\Pi_{k}^\bullet$, respectively.
If $j=k$, the space is 1-dimensional and spanned by $\O^j$.
If $j>k$, the space is empty.
\end{prp}

\pfstart
As in \rf{thm:Johnson}, we have that $\cA_j^k$ is equal to the image of $\grC^-_{t_j}\grS^+_{[2..n-2j]}$ in $\Johnson nk$.
For $T\in {\binom{[n]}k}$, the vector $\grC^-_{t_j}\grS^+_{[2..n-2j]} T$ is a scalar multiple (possibly, zero) of either $\I^k$ or $\O^k$ in dependence on whether $1\in T$ or not, respectively.
\pfend

We have the following estimates on the norms of these vectors:
\[
\normA|\O^k|^2 = 2^j \binom{n-2j-1}{k-j}
\qqand
\normA|\I^k|^2 = 2^j \binom{n-2j-1}{k-j-1}.
\]

We will establish relation between this basis and the $S_n$-isotypical structure of $\Johnson nk$ because the latter is what the morphism $\Gamma$ is using.
Let us rewrite the decomposition~\rf{eqn:JohnsonIrreps} as
\[
\Johnson nk = \cB_0^k\oplus \cB_1^k \oplus\cdots\oplus \cB_k^k,
\]
where $\cB_i^k$ is isomorphic to $\Spe{(n-i,i)}$ as an $S_n$-module.
Now our goal is to find vectors $\widetilde{v_{j,1}}^k \in \cA_j^k\cap \cB_{j}^k$ and $\widetilde{v_{j,2}}^k \in \cA_j^k\cap \cB_{j+1}^k$.

Let us start with the first vector.
If $k=j$, we have
\[
\widetilde {v_{j,1}}^j = \O^j = C_j \in \cA_j^j \cap \cB_j^j.
\]
For $k>j$, we define
\begin{equation}
\label{eqn:wvj1}
\widetilde {v_{j,1}}^k = W^{j\to k} \widetilde {v_{j,1}}^j = C_j\sqtimes R^{[n-2j]}_{k-j}= \O^k + \I^k \in \cA_j^k \cap \cB_j^k,
\end{equation}
where $W^{j\to k}$ is as in~\rf{eqn:W_ell^k}, and we use \rf{eqn:WIsotypicalA}.
The inclusion holds here because $W^{j\to k}$ is an $S_n$-morphism, hence, also an $S_{[2..n]}$-morphism, therefore, it preserves membership in both $\cA_j$ and $\cB_j$.

For the second vector, observe that, for $c\in [2..n-2j]$, the vector
$
C_j\sqtimes \sA[\{c\}- \{1\}]
$
is in $\cB_{j+1}^{j+1}$ by \rf{thm:Johnson}.
Summing over all $c$, we obtain
\[
\widetilde {v_{j,2}}^{j+1} = \sum_{c=2}^{n-2j} C_j\sqtimes \sA[\{c\}- \{1\}] = \O^{j+1} - (n-2j-1) \I^{j+1} \in \cA_j^{j+1} \cap \cB_{j+1}^{j+1}.
\]
Using \rf{clm:Waction}, we get:
\[
W^{j+1\to k}\O^{j+1} = (k-j) \O^k + (k-j-1) \I^k
\qqand
W^{j+1\to k}\I^{j+1} = \I^k.
\]
Hence, we obtain the following vector
\begin{equation}
\label{eqn:wvj2}
\widetilde {v_{j,2}}^k = W^{j+1\to k} \widetilde {v_{j,2}}^{j+1} 
= (k-j) \O^k - (n-j-k) \I^k
\in \cA_j^k \cap \cB_{j+1}^k.
\end{equation}

The normalised versions of these vectors
\[
v_{j,1}^k = \frac{\widetilde {v_{j,1}}^k}{\normA|\widetilde {v_{j,1}}^k|}
\qqand
v_{j,2}^k = \frac{\widetilde {v_{j,2}}^k}{\normA|\widetilde {v_{j,2}}^k|}
\]
form an orthonormal basis of $\cA_j^k$.
Also, by~\rf{eqn:WIsotypicalB}, $W^{\ell\to k}$ is a positive multiple of $\Phi_j^{\ell\to k}$ on $\Spe{(n-j,j)}$, which means that
\[
\Phi_j^{\ell\to k} v_{j,1}^\ell = v_{j,1}^k
\qqand
\Phi_{j+1}^{\ell\to k} v_{j,2}^\ell = v_{j,2}^k.
\]
Therefore, we can write the following isotypical decomposition of $\Gamma$ as an $S_{[2..n]}$-morphism in the sense of~\rf{eqn:isotypical}:
\begin{equation}
\label{eqn:GammaIsotypical}
\Gamma = \bigoplus_j \begin{pmatrix}\gamma_j &  \\  & \gamma_{j+1}  \end{pmatrix} \otimes I_{\Spe{(n-j-1,j)}},
\end{equation}
where the column basis is $v_{j,1}^\ell, v_{j,2}^\ell$, and the row basis is $v_{j,1}^k, v_{j,2}^k$.

The operator $\Pi^\circ_k$ restricted to $\cA_j^k$ is an orthogonal projector onto the vector $\O^k$.
Let us find the coordinates
\[
\phi^{k,\circ}_j = \begin{pmatrix} \phi^{k,\circ}_{j,1} \\ \phi^{k,\circ}_{j,2} \end{pmatrix}
\]
of the normalised $\O^k$ in the basis $\sfigA{v^k_{j,1}, v^k_{j,2}}$.
For that, it suffices to find the inner product between the corresponding normalised vectors.  For the first coordinate, we have
\[
\phi^{k,\circ}_{j,1} = \frac{\ip<\widetilde{v_{j,1}}^k , \O^k>}{\normA|\widetilde{v_{j,1}}^k| \normA|\O^k|} 
=\frac{\normA|\O^k|}{\normA|\widetilde{v_{j,1}}^k|}
=\sqrt{\frac{2^j\binom{n-2j-1}{k-j}} {2^j\binom{n-2j}{k-j} }} 
= \sqrt{ \frac{n-j-k}{n-2j} }.
\]
We can also explicitly calculate the second coordinate.
First, 
\[
\|\widetilde {v_{j,2}}^{j+1}\|^2 
= \|\O^{j+1}\|^2 + (n-2j-1)^2 \|\I^{j+1}\|^2
= 2^j (n-2j-1) + 2^j (n-2j-1)^2 = 2^j(n-2j)(n-2j-1).
\]
Then, using~\rf{eqn:WIsotypicalB} and that $\widetilde {v_{j,2}}^{j+1}\in\cB_{j+1}$:
\[
\|\widetilde {v_{j,2}}^{k}\|^2
= \|W^{j+1\to k} \widetilde {v_{j,2}}^{j+1}\|^2
= \binom{n-2j-2}{k-j-1}\cdot 2^j (n-2j)(n-2j-1).
\]
Which gives us
\[
\phi^{k,\circ}_{j,2} = \frac{\ip<\widetilde{v_{j,2}}^k , \O^k>}{\normA|\widetilde{v_{j,2}}^k| \normA|\O^k|} 
=\frac{(k-j) \normA|\O^k|}{\normA|\widetilde{v_{j,2}}^k|}
=(k-j) \sqrt{\frac{2^j \binom{n-2j-1}{k-j}} {2^j\binom{n-2j-2}{k-j-1}(n-2j)(n-2j-1) }} 
= \sqrt{ \frac{k-j}{n-2j} }.
\]
Thus, the restrictions of $\Pi^\circ_k$ and $\Pi^\bullet_k$ to $\cA_j^k$ project onto the following orthonormal vectors
\[
\phi^{k,\circ}_j = \frac{1}{\sqrt{n-2j}} \begin{pmatrix} \sqrt{n-k-j} \\ \sqrt{k-j} \end{pmatrix}
\qqand
\phi^{k,\bullet}_j = \frac{1}{\sqrt{n-2j}} \begin{pmatrix}\sqrt{k-j} \\ -\sqrt{n-k-j} \end{pmatrix},
\]
respectively, where we use the basis $\sfigA{v^k_{j,1}, v^k_{j,2}}$ as before.
Utilising~\rf{eqn:GammaIsotypical}, we get
\[
\Pi_{k}^\circ \Gamma \Pi_{\ell}^\bullet 
= \bigoplus_j \skB[{\phi^{k,\circ}_j(\phi^{k,\circ}_j)^* 
\begin{pmatrix}\gamma_j &  \\  & \gamma_{j+1}  \end{pmatrix} \phi^{\ell,\bullet}_j(\phi^{\ell,\bullet}_j)^*}] \otimes I_{\Spe{(n-j-1,j)}}.
\]
Therefore,
\[
\norm|\Pi_{k}^\circ \Gamma \Pi_{\ell}^\bullet|
=
\max_j \abs|(\phi^{k,\circ}_j)^* {\begin{pmatrix}\gamma_j &  \\  & \gamma_{j+1}  \end{pmatrix}} \phi^{\ell,\bullet}_j| 
=
\max_j \frac{\abs|\gamma_j \sqrt{(n-k-j)(\ell-j)} - \gamma_{j+1}\sqrt{(k-j)(n-\ell-j)}|} {n-2j}.
\]
Similarly,
\[
\norm|\Pi_{k}^\bullet \Gamma \Pi_{\ell}^\circ|
=
\max_j \abs|(\phi^{k,\bullet}_j)^* {\begin{pmatrix}\gamma_j &  \\  & \gamma_{j+1}  \end{pmatrix}} \phi^{\ell,\circ}_j| 
=
\max_j \frac{\abs|\gamma_j \sqrt{(k-j)(n-\ell-j)} - \gamma_{j+1}\sqrt{(n-k-j)(\ell-j)}|} {n-2j}.
\]
By~\rf{eqn:GammaDelta1}, the norm of $\Gamma\circ\Delta_1$ is the maximum of the two, which gives us \rf{lem:DeltaINorm}.

\section{Proofs of Lemmata~\ref{lem:CoeffsUnderA}--\ref{lem:DeltaReflNorm} }
\label{sec:FirstProofs}

Here we will prove the first three of the main lemmata from \rf{sec:adversaryGeneral}.
For greater clarity, we use $e_x$ to denote the element of the standard basis of $\Johnson nk$ corresponding to $x\subseteq [n]$ in this section.

In the context of Lemmata~\ref{lem:CoeffsUnderA}--\ref{lem:DeltaReflNorm}, the $S_n$-module $\Bigmodule nk$ becomes important.
It is a tensor product of two $S_n$-modules $\Johnson nk$ and $\bC^n = \Johnson n1$.
It can be also seen as an $S_n$-module corresponding to the set $\binom{[n]}k \times [n]$ with the obvious group action of $S_n$ on it.

We will express all the operators appearing in the above lemmata using the following isometry 
\begin{equation}
\label{eqn:Vk}
V_k\colon \Johnson nk \to \Bigmodule nk, 
\quad
e_x \mapsto e_x\otimes \psi_x,
\end{equation}
where $\psi_x = \frac1{\sqrt{|x|}}\sum_{i\in x} \ket |i>$ is as defined in~\rf{eqn:psix}.
In other words, $V_k[(x,i), x] = 1_{i\in x}/\sqrt{k}$, from which \rf{prp:groupAction} tells us $V_k$ is a morphism.

\begin{prp}
\label{prp:Alter}
For an $\binom{[n]}{k}\times \binom{[n]}{\ell}$-matrix $\Gamma$, we have
\begin{align}
\Gamma \circ \Psi &= V_k^* (\Gamma\otimes I_n) V_\ell \label{eqn:PsiAlter},\\
\Gamma \circ \Dp &= V_k\Gamma - (\Gamma\otimes I_n) V_\ell \label{eqn:DpAlter}, \\
\Gamma \circ \Dps &= V_k^*(\Gamma\otimes I_n) - \Gamma V_\ell^*,  \label{eqn:DpsAlter}\\
\Gamma \circ \Dpp &= V_kV_k^* (\Gamma\otimes I_n) - (\Gamma\otimes I_n) V_\ell V_\ell^*, \label{eqn:DppAlter}
\end{align}
where $\Psi_{x,y} = \ip<\psi_x,\psi_y>$, $\Dp_{x,y} = \psi_x - \psi_y$, $\Dps_{x,y} = \psi_x^* - \psi_y^*$, and $\Dpp_{x,y} = \psi_x\psi_x^* - \psi_y\psi_y^*$ are $\binom{[n]}{k}\times \binom{[n]}{\ell}$ (block) matrices defined as in~\rf{eqn:Psi}, \rf{eqn:Deltapsi} and~\rf{eqn:Deltapsipsi}.
Also, $I_n$ stands for the identity on $\bC^n$.
\end{prp}

\pfstart
Everywhere in this proof $x$ ranges over $\binom{[n]}{k}$ and $y$ over $\binom{[n]}{\ell}$.
Let us denote $\gamma_{x,y} = \Gamma\elem[x,y]$ so that 
$\Gamma = 
\sum_{x,y} \gamma_{x,y} e_xe_y^*
$.
Also, we have 
$V_k = 
\sum_x (e_x\otimes \psi_x) e_x^*
$ and 
$V_\ell = 
\sum_y (e_y\otimes \psi_y) e_y^*
$.
Thus,
\[
V_k^* (\Gamma\otimes I_n) V_\ell
= \sum_x e_x (e_x\otimes \psi_x)^*
\sum_{x,y} \gamma_{x,y} (e_xe_y^*\otimes I_n)
\sum_y (e_y\otimes \psi_y) e_y^*
= \sum_{x,y} \gamma_{x,y} e_xe_y^*\cdot \psi_x^*\psi_y
= \Gamma\circ\Psi.
\]
Also,
\[
V_k\Gamma 
= \sum_x (e_x\otimes \psi_x) e_x^* \sum_{x,y} \gamma_{x,y} e_xe_y^* 
= \sum_{x,y} \gamma_{x,y} (e_x\otimes \psi_x) e_y^*
\]
and
\[
(\Gamma\otimes I_n) V_\ell
= \sum_{x,y} \gamma_{x,y} (e_xe_y^*\otimes I_n) \sum_y (e_y\otimes \psi_y) e_y^*
= \sum_{x,y} \gamma_{x,y} (e_x\otimes \psi_y) e_y^*,
\]
from which~\rf{eqn:DpAlter} follows since
\[
\Gamma\circ\Dp 
= \sum_{x,y} \gamma_{x,y} \sA[e_x\otimes (\psi_x - \psi_y)] e_y^*.
\]
Similarly,
\[
V_k^*(\Gamma\otimes I_n)
= \sum_x e_x (e_x\otimes \psi_x)^* \sum_{x,y} \gamma_{x,y} (e_xe_y^*\otimes I_n)
= \sum_{x,y} \gamma_{x,y} e_x(e_y\otimes \psi_x)^*
\]
and
\[
\Gamma V_\ell^*
= \sum_{x,y} \gamma_{x,y} e_xe_y^* \sum_y e_y (e_y\otimes \psi_y)^*
= \sum_{x,y} \gamma_{x,y} e_x(e_y\otimes \psi_y)^*,
\]
which implies~\rf{eqn:DpsAlter}.
Finally,
\[
V_kV_k^* (\Gamma\otimes I_n) 
= \sum_x (e_x\otimes\psi_x)(e_x\otimes\psi_x)^* \sum_{x,y} \gamma_{x,y} (e_xe_y^*\otimes I_n)
= \sum_{x,y} \gamma_{x,y} e_xe_y^*\otimes \psi_x\psi_x^*
\]
and
\[
(\Gamma\otimes I_n) V_\ell V_\ell^*
= \sum_{x,y} \gamma_{x,y} (e_xe_y^*\otimes I_n) \sum_y (e_y\otimes\psi_y)(e_y\otimes\psi_y)^*
= \sum_{x,y} \gamma_{x,y} e_xe_y^*\otimes \psi_y\psi_y^*,
\]
from which we get~\rf{eqn:DppAlter}.
\pfend

Let us explore the structure of $\Bigmodule nk$.
We know from \rf{thm:Johnson} that $\Johnson nk$ decomposes into irreducible submodules as
\[
\Johnson nk = \cB_0^k\oplus \cB_1^k \oplus\cdots\oplus \cB_k^k,
\]
where $\cB_i^k$ is isomorphic to $\Spe{(n-i,i)}$.
Similarly, as a special case of $k=1$, we have the following decomposition into submodules:
\[
\bC^n = \cE_0 \oplus \cE_1
\]
with $\cE_i$ is isomorphic to $\Spe{(n-i,i)}$.
We are interested in the relation between tensor products of these submodules and the isotypical decomposition of $\Bigmodule nk$.
We will prove the following two results in \rf{sec:bigModule}.

\begin{lem}
\label{lem:basis}
Assume $j>0$ and $k>j$.
For the vector $C_j\in \Spe{(n-j,j)}$ from~\rf{eqn:C_j}, the space
\begin{equation}
\label{eqn:A}
\cA_j^k = \Slice_{S_n}\sA[C_j, \Bigmodule nk]
\end{equation}
is 4-dimensional with an orthonormal basis $w_{j,1}^{k}, w_{j,2}^{k}, w_{j,3}^{k}, w_{j,4}^{k}$, where
\begin{equation}
\label{eqn:v_j0k}
w_{j,1}^{k} \in \cB_{j-1}^{k}\otimes\cE_1,\quad
w_{j,2}^{k} \in \cB_{j}^{k}\otimes\cE_0,\quad
w_{j,3}^{k} \in \cB_{j}^{k}\otimes\cE_1,\quad\text{and}\quad
w_{j,4}^{k} \in \cB_{j+1}^{k}\otimes\cE_1.
\end{equation}
If $j=0$, the space is 2-dimensional and spanned by $w_{j,2}^{k}$ and $w_{j,4}^{k}$.
We additionally have, for all $\ell$ and $k$ where the corresponding operators exist:
\begin{equation}
\label{eqn:PhiNaBaze}
\begin{aligned}
(\Phi_{j-1}^{\ell\to k}\otimes I_n)w_{j,1}^\ell &= w_{j,1}^{k},\qquad&
(\Phi_{j}^{\ell\to k}\otimes I_n)w_{j,2}^\ell &= w_{j,2}^{k},\qquad\\
(\Phi_{j}^{\ell\to k}\otimes I_n)w_{j,3}^\ell &= w_{j,3}^{k},&
(\Phi_{j+1}^{\ell\to k}\otimes I_n)w_{j,4}^\ell &= w_{j,4}^{k}.
\end{aligned}
\end{equation}
\end{lem}

\begin{lem}
\label{lem:operatorV}
For the operator $V_k$ defined in~\rf{eqn:Vk}, we have the following isotypical decomposition in the sense of~\rf{eqn:isotypical}:
\begin{equation}
\label{eqn:VkDecompose}
V_k = \bigoplus_j \fita_j^k \otimes I_{\Spe{(n-j,j)}}
 = \bigoplus_j 
\begin{pmatrix} \fita^k_{j,1} \\ \fita^k_{j,2} \\ \fita^k_{j,3} \\ \fita^k_{j,4} \end{pmatrix}
\otimes I_{\Spe{(n-j,j)}},
\end{equation}
where $\fita^k_j$ are defined in~\rf{eqn:fita}.
The column basis is $v_j^k$ from~\rf{eqn:v_jk} and the row basis is $w_{j,1}^k, \dots, w_{j,4}^k$ from~\rf{eqn:v_j0k}.
\end{lem}

Now we are in position to prove the lemmata from~\rf{sec:adversaryGeneral}.
We prove them for an $\binom{[n]}k\times \binom{[n]}\ell$ matrix $\Gamma$ given by
\begin{equation}
\label{eqn:GammaDecompose}
\Gamma = \bigoplus_j \gamma_j \Phi_j^{\ell\to k}
= \bigoplus_j (\gamma_j)\otimes I_{\Spe{(n-j,j)}},
\end{equation}
where the last one is its isotypical decomposition in the sense of~\rf{eqn:isotypical}.
In the latter, we assume that the $1\times 1$-matrix has the column basis $v_j^\ell$ and the row basis $v_j^k$.
Similarly, using~\rf{eqn:PhiNaBaze}:
\begin{equation}
\label{eqn:GammaIDecompose}
\Gamma\otimes I_n 
= \bigoplus_j \Gamma_j \otimes I_{\Spe{(n-j,j)}}
= \bigoplus_j 
\begin{pmatrix}
\gamma_{j-1} & & &\\
& \gamma_j & & \\
&& \gamma_j &\\
&&& \gamma_{j+1}
\end{pmatrix}
\otimes I_{\Spe{(n-j,j)}},
\end{equation}
where the notation $\Gamma_j$ is borrowed from~\rf{eqn:fita}.
Here the column basis is $w_{j,1}^\ell, \dots, w_{j,4}^\ell$ and the row basis is $w_{j,1}^k, \dots, w_{j,4}^k$.

Now, by \rf{prp:Alter} and using the decompositions~\rf{eqn:VkDecompose},~\rf{eqn:GammaDecompose} and~\rf{eqn:GammaIDecompose} above, we have
\[
\begin{alignedat}{3}
\Gamma \circ \Psi 
&= V_k^* (\Gamma\otimes I_n) V_\ell
&&= \bigoplus_j 
\sA[(\fita_{j}^k)^* \Gamma_j \fita_j^\ell] \otimes I_{\Spe{(n-j,j)}}.
\\
\Gamma \circ \Dp 
&= V_k\Gamma - (\Gamma\otimes I_n) V_\ell 
&&=
\bigoplus_j (\gamma_j\fita^k_j - \Gamma_j\fita^{\ell}_j) \otimes I_{\Spe{(n-j,j)}}.
\\
\Gamma \circ \Dps 
&= V_k^*(\Gamma\otimes I_n) - \Gamma V_\ell^*
&&= 
\bigoplus_j (\Gamma_j\fita_j^k - \gamma_j\fita_j^\ell)^* \otimes I_{\Spe{(n-j,j)}}.\\
\Gamma \circ \Dpp 
&= V_kV_k^* (\Gamma\otimes I_n) - (\Gamma\otimes I_n) V_\ell V_\ell^*
&& = \bigoplus_j \sA[ \fita_j^k (\fita_j^k)^*\Gamma_j  - \Gamma_j \fita_j^\ell(\fita_j^\ell)^* ]\otimes I_{\Spe{(n-j,j)}} .
\end{alignedat}
\]
Lemmata~\ref{lem:CoeffsUnderA}, \ref{lem:DeltaGenNorm} and~\ref{lem:DeltaReflNorm} follow from the above four equations with $\ell=k'$ and with the use of~\rf{eqn:isotypicalNorm} to bound the norms where necessary.

\section{The Module \texorpdfstring{$\Bigmodule nk$}{C([n] choose k) times Cn}}
\label{sec:bigModule}

In this section, we analyse the $S_n$-module $\Bigmodule nk$ similarly as we did for the $S_n$-module $\Johnson nk$ in \rf{sec:AssocPrelim}.
We will denote the basis elements of this module by $T\otimes \{d\}$ with $T\in \binom{[n]}k$ and $d\in[n]$.


%

\mycutecommand{\grD}{\EuScript{D}}

By \rf{thm:Et}, to get a slice of $\Bigmodule nk$, it suffices to apply $\grC^-_{t_j}\grR^+_{t_j}$ with $t_j$ in~\rf{eqn:t_j}.
However, directly dealing with this group algebra element is complicated.
We will use the following group algebra element instead:
\begin{equation}
\label{eqn:Dj}
\grD_j = \grC^-_{t_j} \grS_{[n-2j]}^+ \grS_{(a_1,b_1),\dots,(a_j,b_j)}^+,
\end{equation}
where $S_{(a_1,b_1),\dots,(a_j,b_j)}$ is the group permuting the pairs $\{{(a_1,b_1),\dots,(a_j,b_j)}\}$.
The latter is isomorphic to $S_j$, where $\sigma\in S_j$ corresponds to the permutation mapping $a_i \mapsto a_{\sigma(i)}$, $b_i\mapsto b_{\sigma(i)}$ for $i\in[j]$, and the identity elsewhere.

\begin{prp}
\label{prp:Dj}
The operator $\grD_j$ is a scalar times an orthogonal projector.
On every $S_n$-module $\cM$, its image contains the image of $\grC^-_{t_j}\grR_{t_j}^+$.
Moreover, if $\cM\cong\Spe{(n-i,i)}$, the two images coincide.
In particular, $\grD_j$ annihilates all $\Spe{(n-i,i)}$ with $i\ne j$.
\end{prp}

\begin{proof}
First, all the multipliers in~\rf{eqn:Dj} are orthogonal projectors up to a constant, and they commute.
Hence, $\grD_j$ is also a constant times an orthogonal projector.

Both $S_{[n-2j]}$ and $S_{(a_1,b_1),\dots,(a_j,b_j)}$ are subgroups of $R_{t_j}$, hence
\[
(n-2j)!j! \grC^-_{t_j}\grR^+_{t_j} =
\grC^-_{t_j} \grS_{[n-2j]}^+\grS_{(a_1,b_1),\dots,(a_j,b_j)}^+ \grR^+_{t_j} = \grD_j\grR^+_{t_j}.
\]
Thus, the image of $\GrProjector$ is contained in the image of $\grD_j$.

Also, $\grD_j$ contains $\grC^-_{t_j}\grS_{[n-2j]}^+$, hence, by~\rf{thm:Johnson}, the image of $\grD_j$ is contained in the image of $\GrProjector$ in all $\Spe\lambda$ with $\lambda$ having 2 or less rows.
\end{proof}

One problem is that $\grD_j$ is not guaranteed to annihilate irreps $\Spe\lambda$ with $\lambda$ having more than 2 rows.
Because of that, the image of $\grD_j$ turns out to be strictly larger than that of $\GrProjector$, which annihilates all such irreps.

The main reason for our choice of $\grD_j$ is that all the multipliers in~\rf{eqn:Dj} commute, therefore, it is very easy to describe its image.
For that, we will need the following version of the vector $C_j$ from~\rf{eqn:C_j} with one of the multipliers dropped:
\newcommand{\Cs}{C^{{\curvearrowright}{s}}_{j-1}}
\[
\Cs = 
\sB[\{a_1\}-\{b_1\}]\sqtimes\cdots \sqtimes \sB[\{a_{s-1}\}-\{b_{s-1}\}]\sqtimes\sB[\{a_{s+1}\}-\{b_{s+1}\}]\sqtimes \cdots\sqtimes \sB[\{a_j\}-\{b_j\}].
\]

\renewcommand{\O}{{[\mathrm{O}]}_j}
\newcommand{\IA}{{[\mathrm{I}_{\mathrm A}]}_j}
\newcommand{\IB}{{[\mathrm{I}_{\mathrm B}]}_j}
\newcommand{\II}{{[\mathrm{II}]}_j}
\newcommand{\IIIA}{{[\mathrm{III}_{\mathrm A}]}_j}
\newcommand{\IIIB}{{[\mathrm{III}_{\mathrm B}]}_j}

\begin{prp}
\label{prp:basis}
Assuming $k\ge j+1$, the image of $\grD_j$ in $\Bigmodule nk$ has the following orthogonal basis:
\begin{align*}
\O^k &= \sum_{s=1}^j \Cs \sqtimes R^{[n-2j]}_{k-j+1} \otimes (\{a_s\}-\{b_s\}) ,\\
\IA^k &= \sum_{s=1}^j \Cs \sqtimes R^{[n-2j]}_{k-j} \sqtimes (\{a_s\}\otimes\{a_s\}-\{b_s\}\otimes\{b_s\}),\\
\IB^k &= \sum_{s=1}^j \Cs \sqtimes R^{[n-2j]}_{k-j} \sqtimes (\{a_s\}\otimes\{b_s\}-\{b_s\}\otimes\{a_s\}) ,\\
\II^k &= \sum_{s=1}^j \Cs \sqtimes \{a_s,b_s\} \sqtimes R^{[n-2j]}_{k-j-1} \otimes (\{a_s\}-\{b_s\}) ,\\
\IIIA^k &= C_j \sqtimes \sum_{T\subseteq [n-2j], |T|=k-j} T\otimes \sum_{d\in T} \{d\},\\
\IIIB^k &= C_j \sqtimes \sum_{T\subseteq [n-2j], |T|=k-j} T\otimes \sum_{d\in [n-2j]\setminus T} \{d\}.\\
\end{align*}
If $k=j$, the basis consists of $\O^j$, $\IA^j$, $\IB^j$ and $\IIIB^j$.
If $k=j-1$, the basis consists only of $\O^{j-1}$.
For $k<j-1$, the image is zero.
\end{prp}

\pfstart
Take any basis element $T\otimes \{d\}$ of $\Bigmodule nk$ and consider its image under $\grD_j$.
If $d\in[n-2j]$, the derivation is similar to that of~\rf{eqn:R^A_ell} and the image is a scalar multiple of either $\IIIA^k$ or $\IIIB^k$ in dependence on whether $d\in T$ or not.

Now consider the case $d>n-2j$.  We may assume $d=a_j$, the other cases being similar.
Again, we have that $\grD_j T = 0$ unless $\absA|T\cap \{a_i, b_i\}|=1$ for all $i<j$, so we assume the latter condition is satisfied.
For $T\cap\{a_j,b_j\}$ all four options $\emptyset$, $a_j$, $b_j$, or $\{a_j,b_j\}$ are permissible though, and we get that $\grD_j T$ a scalar multiple of $\O^k$, $\IA^k$, $\IB^k$, or $\II^k$,  respectively.

The vectors are orthogonal because they are supported on pairwise disjoint subsets of basis vectors.
The cases of $k=j$ and $k=j-1$ follow from observing that, in these cases, the unlisted vectors are 0.
\pfend

We will need few properties of these vectors.
First, simple calculation reveals that the norms of these vectors satisfy
\begin{align*}
\normB|\O^k|^2 &= j2^j {n-2j\choose k-j+1}&\qquad
\normB|\IA^k|^2 &= \normB|\IB^k|^2 = j2^j {n-2j\choose k-j}\\
\normB|\II^k|^2 &= j2^j {n-2j\choose k-j-1}\\
\normB|\IIIA^k|^2 &= (k-j)2^j {n-2j\choose k-j}&\qquad
\normB|\IIIB^k|^2 &= (n-k-j)2^j {n-2j\choose k-j}.
\end{align*}
Second, using \rf{clm:Waction}, we get the following action of $W^{\ell\to k}\otimes I_n$ on these vectors, where $I_n$ acts on $\bC^n$:
\begin{equation}
\label{eqn:WOnBasis}
\begin{aligned}
\O^\ell &\mapsto {k-j+1\choose \ell-j+1} \O^k + {k-j\choose \ell-j+1} \sA[\IA^k - \IB^k] + {k-j-1\choose \ell-j+1} \II^k,\\
\IA^\ell &\mapsto {k-j\choose \ell-j}\IA^k + {k-j-1\choose \ell-j}\II^k,\\
\IB^\ell &\mapsto {k-j\choose \ell-j}\IB^k - {k-j-1\choose \ell-j}\II^k,\\
\II^\ell &\mapsto {k-j-1\choose \ell-j-1} \II^k,\\
\IIIA^\ell &\mapsto {k-j-1\choose \ell-j-1} \IIIA^k,\\
\IIIB^\ell &\mapsto {k-j-1\choose \ell-j}\IIIA^k + {k-j\choose \ell-j}\IIIB^k.
\end{aligned}
\end{equation}


\subsection{Proof of \rf{lem:basis}}
\label{sec:basis}

\myfigure{\label{fig:proof}}
{
The structure of the image of $\grD_j$ in various modules.
The columns correspond to the modules $\Bigmodule n\ell$ for various $\ell$ starting with $\ell=j-1$.
If $\ell>j$, the image of $\grD_j$ in this module is 6-dimensional, but only 4 of these dimensions are in $\cA_j^\ell$ (highlighted in yellow).
We find the first vector in each row (indicated by the framed vectors), and use the morphism $W^{\ell\to k}$ to move further in each row (indicated by the arrows).
The bullets and the stars correspond to vectors perpendicular to $\cA_j^\ell$.
They are identified in~\rf{eqn:triple j otimes 1} and~\rf{eqn:tripe j+1 otimes 1}.
}
{
\negbigskip
\negbigskip
\[
\def\nLabel{m-\the\pgfmatrixcurrentrow-\the\pgfmatrixcurrentcolumn}
\def\nB{\node[nB](\nLabel){}; \node}
\def\nY{\node[nB, nY](\nLabel){}; \node}
\def\nG{\node[nB](\nLabel){$\bullet$}}
\def\nH{\node[nB](\nLabel){$*$}}
\begin{tikzpicture}[
    nA/.style={minimum height=0.8cm},
    nB/.style={nA, minimum width=2.5cm},
    nD/.style={nA, minimum width=1.5cm},
    nY/.style={fill=yellow!50},
    nBrace/.style={decorate,decoration={brace,amplitude=0.2cm,mirror}},
    tBrace/.style={left=0.2cm,align=right}
]
\matrix[ampersand replacement=\&]
{
\nB{$\Bigmodule n{j-1}$}; \& \nB{$\Bigmodule n{j}$}; \& \nB{$\Bigmodule n{j+1}$}; \& \node[nD]{$\cdots$}; \& \nB{$\Bigmodule n{k}$}; \& \node[nD]{$\cdots$};\\
\nY[draw](A1){$\widetilde{w_{j,1}}^{j-1}$}; \& \nY(A2){$\widetilde{w_{j,1}}^{j}$}; \& \nY(A3){$\widetilde{w_{j,1}}^{j+1}$}; \& \node[nD,nY]{$\cdots$}; \& \nY(A4){$\widetilde{w_{j,1}}^{k}$}; \& \node[nD,nY]{$\cdots$};\\
\& \nY[draw](B2){$\widetilde{w_{j,2}}^{j}$}; \& \nY(B3){$\widetilde{w_{j,2}}^{j+1}$}; \& \node[nD,nY]{$\cdots$}; \& \nY(B4){$\widetilde{w_{j,2}}^{k}$}; \& \node[nD,nY]{$\cdots$};\\
\& \nG; \& \nG; \& \node[nD]{$\cdots$}; \& \nG; \& \node[nD]{$\cdots$};\\
\& \nY[draw](C2){$\widetilde{w_{j,3}}^{j}$}; \& \nY(C3){$\widetilde{w_{j,3}}^{j+1}$}; \& \node[nD,nY]{$\cdots$}; \& \nY(C4){$\widetilde{w_{j,3}}^{k}$}; \& \node[nD,nY]{$\cdots$};\\
\&\& \nH; \& \node[nD]{$\cdots$}; \& \nH; \& \node[nD]{$\cdots$};\\
\&\& \nY[draw](D3){$\widetilde{w_{j,4}}^{j+1}$}; \& \node[nD,nY]{$\cdots$}; \& \nY(D4){$\widetilde{w_{j,4}}^{k}$}; \& \node[nD,nY]{$\cdots$};\\
};
\draw[->] (A1) to (A2);
\draw[->, bend left=9] (A1) to (A3);
\draw[->, bend left=11] (A1) to (A4);
\draw[->] (B2) to (B3);
\draw[->, bend left=10] (B2) to (B4);
\draw[->] (C2) to (C3);
\draw[->, bend left=10] (C2) to (C4);
\draw[->, bend left=5] (D3) to (D4);
\draw (m-2-1.north west) rectangle (m-2-1.south east);
\draw (m-2-2.north west) rectangle (m-5-2.south east);
\draw (m-2-3.north west) rectangle (m-7-3.south east);
\draw (m-2-5.north west) rectangle (m-7-5.south east);
\draw [nBrace] (m-2-1.north west) --node[tBrace]{$\cB_{j-1}^\ell\otimes\cE_1$} (m-2-1.south west);
\draw [nBrace] (m-3-2.north west) --node[tBrace]{$\cB_{j}^\ell\otimes\cE_0$} (m-3-2.south west);
\draw [nBrace] (m-4-2.north west) --node[tBrace]{$\cB_{j}^\ell\otimes\cE_1$} (m-5-2.south west);
\draw [nBrace] (m-6-3.north west) --node[tBrace]{$\cB_{j+1}^\ell\otimes\cE_1$} (m-7-3.south west);
\end{tikzpicture}
\]
\negmedskip
}

Our overall proof strategy is as follows, see also \rf{fig:proof}.
First, we will identify non-zero vectors
\begin{align*}
&&
\widetilde{w_{j,2}}^j &\in \cA_j^j\cap\sA[\cB_{j}^{j}\otimes\cE_0], &
\\
\widetilde{w_{j,1}}^{j-1} &\in \cA_j^{j-1}\cap\sA[\cB_{j-1}^{j-1}\otimes\cE_1],&
\widetilde{w_{j,3}}^j &\in \cA_j^j\cap\sA[\cB_{j}^{j}\otimes\cE_1], &
\widetilde{w_{j,4}}^{j+1} &\in \cA_j^{j+1}\cap\sA[\cB_{j+1}^{j+1}\otimes\cE_1]
\end{align*}
which feature the smallest values of $k$, where the corresponding spaces are non-zero.
Let, correspondingly, $\ell_{j,1} = j-1$, $\ell_{j,2}=\ell_{j,3} = j$, and $\ell_{j,4} = j+1$.
Then, we define
\[
w^k_{j,i} = \frac{
\widetilde{w_{j,i}}^k 
} {\normA|
\widetilde{w_{j,i}}^k
|}
\qquad\text{with}\qquad
\widetilde{w_{j,i}}^k = \sA[W^{\ell_{j,i}\to k}\otimes I_n] \widetilde{w_{j,i}}^{\ell_{j,i}}.
\]
The operation $W^{\ell\to k}\otimes I_n$ is well-suited for that since it is a morphism from $\Bigmodule n\ell$ to $\Bigmodule nk$, but it also maps $\cB^{\ell}_{j}\otimes\cE_i$ into $\cB^{k}_{j}\otimes\cE_i$.
Also, by \rf{prp:WIsotypical}, $W^{\ell\to k}$ is a positive multiple of $\Phi^{\ell\to k}_j$ on $\Spe{(n-j,j)}$, hence, Eq.~\rf{eqn:PhiNaBaze} is automatically satisfied.

First, by \rf{prp:basis}, the image $\grD_j\sB[\Bigmodule n{j-1}]$ is 1-dimensional and is spanned by
\[
\widetilde{w_{j,1}}^{j-1} = \O^{j-1} = \sum_{j=1}^s \Cs \otimes (\{a_s\}-\{b_s\}) \in \cA_j^{j-1}\cap\sA[\cB_{j-1}^{j-1}\otimes\cE_1].
\]
For the inclusion, note that each term in the sum belongs to $\cB_{j-1}^{j-1}\otimes\cE_1$, and it is the image of $C_j$ under the morphism $T\mapsto \sum_{d\in T} (T\setminus\{d\})\otimes\{d\}$, hence, belongs to $\cA_j^{j-1}$.

Now consider the image $\grD_j\sB[\Bigmodule n{j}]$, which, by the same \rf{prp:basis}, is 4-dimensional.
We already know one its vector $\widetilde{w_{j,1}}^j \in \cA_j^j\cap \sA[\cB^j_{j-1}\otimes \cE_1]$.
Also,
\begin{equation*}
\widetilde{w_{j,2}}^j = \IA^j + \IB^j + \IIIB^j
= C_j\otimes R^{[n]}_{1}
\in
\cA_j^j\cap \sA[\cB^j_{j}\otimes \cE_0].
\end{equation*}
Indeed, membership in $\cB^j_{j}\otimes \cE_0$ is obvious, and this vector is the image of $C_j$ under the morphism $T\mapsto T\otimes R^{[n]}_{1}$.

The remaining two dimensions of $\grD_j\sB[\Bigmodule n{j}]$ come from $B^j_j\otimes \cE_1$.
It is easier to identify the one orthogonal to $\cA_j^j$.
Consider the vector
\[
\Cs \sqtimes \sB[\{a_s\}\otimes\{b_s\} - \{b_s\}\otimes\{a_s\} - \{a_s\}\otimes\{d\} - \{d\}\otimes\{b_s\} + \{b_s\}\otimes\{d\} + \{d\}\otimes\{a_s\}],
\]
where $d\in [n-2j]$.
It is an eigenvector of $\grS^-_{\{a_s, b_s, d\}}$, 
hence, it is orthogonal to $\cA_j^j$ by \rf{clm:Sabc}.
We can pair the terms in the sum in two different ways: 
\begin{align*}
&\Cs \sqtimes \skB[{\sB[\{a_s\} - \{d\}]\otimes\{b_s\} + \sB[\{d\}- \{b_s\}]\otimes\{a_s\} + \sB[\{b_s\}-\{a_s\}]\otimes\{d\} }]\\
= &\Cs \sqtimes \skB[{\{a_s\}\otimes\sB[\{b_s\}- \{d\}] + \{b_s\}\otimes\sB[\{d\} - \{a_s\}] +  \{d\}\otimes\sB[\{a_s\}-\{b_s\}]}],
\end{align*}
which shows that this vector is in $\cB^j_j\otimes\cE_1$.
By summing over all $s\in[j]$ and $d\in[n-2j]$, we obtain the following vector in the image of $\grD_j$:
\begin{equation}
\label{eqn:triple j otimes 1}
\O^j + (n-2j)\IB^j - j \IIIB^j \in (\cA^j_j)^\perp \cap \sA[\cB_j^j\otimes\cE_1].
\end{equation}
The remaining fourth vector can be obtained by finding the vector orthogonal to~$\widetilde{w_{j,1}}^j$, $\widetilde{w_{j,2}}^j$ and~\rf{eqn:triple j otimes 1}.
However, it is also possible to use the vector
\[
(n-2j)\IA^j +(n-2j)\IB^j -2j \IIIB^j  = C_j\otimes\sC[ (n-2j)\sum_{a=n-2j+1}^n \{a\} - 2j\sum_{d\in [n-2j]} \{d\}  ] \in \cB_j^j\otimes\cE_1
\]
and perform Gram-Schmidt orthogonalisation process with~\rf{eqn:triple j otimes 1}.
This gives us the following vector in $\cB^j_j\otimes\cE_1$ orthogonal to~\rf{eqn:triple j otimes 1}:
\begin{equation}
\label{eqn:j otimes 1}
\widetilde{w_{j,3}}^j = 
- n\O^j + (n-j+1)(n-2j) \IA^j - (j-1)(n-2j)\IB^j - j(n-2j+2) \IIIB^j.
\end{equation}
Now we show that this vector belongs to $\cA_j^j$.
Note that $\IA^j\in A_j^j$ as the image of $C_j$ under the morphism $T\mapsto T\otimes\sum_{d\in T} d$.
The vector $\widetilde{w_{j,3}}^j\in \cB^j_j\otimes\cE_1$ has non-zero inner product with $\IA^j$, hence, $\cB^j_j\otimes\cE_1$ contains a copy of the irrep $\Spe{(n-j,j)}$.
Thus, by \rf{prp:Dj}, $\grD_j(\cB^j_j\otimes\cE_1)$ must contain a vector from $\cA_j^j$.
As the space is 2-dimensional, and the vector in~\rf{eqn:triple j otimes 1} is orthogonal to $\cA_j^j$, the vector $\widetilde{w_{j,3}}^j$, which is orthogonal to~\rf{eqn:triple j otimes 1}, lies in $\cA_j^j$.
Therefore, we have $\widetilde{w_{j,3}}^j \in \cA_j^j\cap\sA[\cB_{j}^{j}\otimes\cE_1]$ as required.


Now let us move to $\grD_j\sB[\Bigmodule n{j+1}]$.
The situation is similar.  The space is 6-dimensional, and we have identified four dimensions.
The remaining two come from $\cB_{j+1}^{j+1}\otimes \cE_1$.
Again, it is easier to identify the one perpendicular to $\cA_{j}^{j+1}$.
Consider the vector
\begin{align*}
\Cs\sqtimes \biggl[&(\{a_s\}-\{b_s\})\sqtimes (\{c\}-\{d\}) \otimes (\{c\}-\{d\}) \\
+ &(\{b_s \}-\{ c \})\sqtimes (\{ a_s \}-\{ d \}) \otimes (\{ a_s \}-\{ d \})\\
+ &(\{ c \}-\{ a_s \})\sqtimes (\{ b_s \}-\{ d \}) \otimes (\{ b_s \}-\{ d \}) \biggr],
\end{align*}
where $c,d\in [n-2j]$ are distinct.
The vector is in $\cB_{j+1}^{j+1}\otimes \cE_1$.
It is an eigenvector or $\grS^-_{\{a_s, b_s, c\}}$, hence, it is orthogonal to  $\cA_{j}^{j+1}$ by \rf{clm:Sabc}.
Summing over $s\in [j]$ and $c,d\in [n-2j]$ gives us the following vector in the image of $\grD_j$:
\begin{equation}
\label{eqn:tripe j+1 otimes 1}
\begin{aligned}
2\O^{j+1} &- (n-2j-1)\IA^{j+1} + (n-2j-1) \IB^{j+1} + (n-2j)(n-2j-1)\II^{j+1} \\
&+ (n-2j-1)j\IIIA^{j+1} - j \IIIB^{j+1} \in (\cA_{j}^{j+1})^\perp \cap \sA[\cB^{j+1}_{j+1}\otimes \cE_1].
\end{aligned}
\end{equation}
In principle, it is possible to find the sixth vector by solving a $5\times 6$ linear system, but there is a more computationally efficient way.
Observe that
\[
C_j\sqtimes (\{c\}-\{d\}) \otimes (\{c\}-\{d\}) \in \cB_{j+1}^{j+1}\otimes\cE_1.
\]
Summing over $c,d\in [n-2j]$ gives
\[
(n-2j-1) \IIIA^{j+1} - \IIIB^{j+1} \in \cB_{j+1}^{j+1}\otimes\cE_1.
\]
Performing Gram–Schmidt orthogonalization on two vectors gives us the following vector in $\cB_{j+1}^{j+1}\otimes \cE_1$ orthogonal to~\rf{eqn:tripe j+1 otimes 1}:
\begin{equation*}
\begin{aligned}
\widetilde{w_{j,4}}^{j+1} &= - 2 \O^{j+1} + (n-2j-1)\IA^{j+1} -(n-2j-1) \IB^{j+1} -(n-2j)(n-2j-1) \II^{j+1} \\ &+(n-2j+1)(n-2j-1) \IIIA^{j+1} - (n-2j+1) \IIIB^{j+1}.
\end{aligned}
\end{equation*}
This vector belongs to $\cA_j^{j+1}$ because it has non-zero inner product with $\IIIA^{j+1} - \II^{j+1}$, which is the image of $C_j$ under the morphism $T\mapsto T\sqtimes \sum_{d\notin T} \{d\}\otimes \{d\}$.

Let us note in passing, although we will not need this, that the vectors from~\rf{eqn:j otimes 1} and~\rf{eqn:tripe j+1 otimes 1} are from the isotypical subspace of $\Spe{(n-j-1,j,1)}$.

\subsection{Proof of \rf{lem:operatorV}}
\label{sec:operatorV}

To find decomposition of $V_k$, it suffices to find the projections of the normalized
\[
u^k_j = \sqrt{k} V_k\sA[C_j\sqtimes R^{[n-2j]}_{k-j}] = \IA^k + \IIIA^k
\]
onto the vectors $w_{j,1}^{k}, w_{j,2}^{k}, w_{j,3}^{k}, w_{j,4}^{k}$ from~\rf{sec:basis}.
We have
\[
\|u_j^k\|^2 = k 2^j {n-2j\choose k-j}.
\]

First, take the vector
\[
\widetilde{w_{j,1}}^k = \sA[W^{j-1\to k}\otimes I_n]\widetilde{w_{j,1}}^{j-1}  = 
\O^k + \IA^k - \IB^k + \II^k .
\]
We have
\[
\ip<u_j^k, \widetilde{w_{j,1}}^k> = \|\IA^k\|^2 = j2^j {n-2j\choose k-j}
\qqand
\normA|\widetilde{w_{j,1}}^k|^2 = j 2^j {n-2j+2\choose k-j+1},
\]
which gives
\[
\fita_{j,1}^k = 
\frac{\ip<u_j^k, \widetilde{w_{j,1}}^k >}{\|u_j^k\|\cdot \|\widetilde{w_{j,1}}^k\|} 
= \sqrt{\frac{j(k-j+1)(n-k-j+1)}{k(n-2j+2)(n-2j+1)}}.
\]

Second, take the vector
\[
\widetilde{w_{j,2}}^k = \sA[W^{j\to k}\otimes I_n]\widetilde{w_{j,1}}^{j}
= C_j\sqtimes R^{[2n-j]}_{k-j} \otimes R^{[n]}_{1}
\]
We have
\[
\ip<u_j^k , \widetilde{w_{j,2}}^k> =  k 2^j {n-2j\choose k-j}
\qqand
\normA|\widetilde{w_{j,2}}^k|^2 = n 2^j {n-2j\choose k-j}
\]
which gives
\[
\fita_{j,2}^k = 
\frac{\ip<u_j^k, \widetilde{w_{j,2}}^k >}{\|u_j^k\|\cdot \|\widetilde{w_{j,2}}^k\|} 
= \sqrt{\frac kn}.
\]

Third, consider the vector
$\widetilde{w_{j,3}}^k = (W^{j\to k}\otimes I_n)\widetilde{w_{j,3}}^j$.
Eq.~\rf{eqn:WOnBasis} and some simplifications give us
\begin{align*}
\ip<u_j^k, \widetilde{w_{j,3}}^k> &= \sA[-(k-j)n+(n-2j)(n-j+1)] \normA|\IA^k|^2 - j(n-2j+2) \normA|\IIIA^k|^2 \\
& = 2^j j(n-j+1)(n-2k) {n-2j\choose k-j}.
\end{align*}
To get the norm of $\widetilde{w_{j,3}}^k$, we can use \rf{eqn:WIsotypicalB}.  After some calculations, we arrive at
\[
\norm|\widetilde{w_{j,3}}^k|^2 = {n-2j\choose k-j}  \|\widetilde{w_{j,3}}^j\|^2 = 2^j n j(n-2j)(n-j+1)(n-2j+2) {n-2j\choose k-j}.
\]
Therefore,
\[
\fita_{j,3}^k = 
\frac{\ip<u_j^k, \widetilde{w_{j,3}}^k >}{\|u_j^k\|\cdot \|\widetilde{w_{j,3}}^k\|} 
= \frac{n-2k}{\sqrt{nk}} \sqrt{\frac{j(n-j+1)}{(n-2j)(n-2j+2)}}.
\]

Finally, consider the vector
$\widetilde{w_{j,4}}^k = (W^{j+1\to k}\otimes I_n)\widetilde{w_{j,4}}^{j+1}$.
Eq.~\rf{eqn:WOnBasis} gives us
\begin{align*}
\ip<u_j^k, \widetilde{w_{j,4}}^k > &= (n-j-k)(k-j) \normA|\IA^k|^2 + (n-j-k)(n-2j+1) \normA|\IIIA^k|^2 \\
& = 2^j (n-j-k)(k-j)(n-j+1) {n-2j\choose k-j}.
\end{align*}
For the norm of $\widetilde{w_{j,4}}^k$, we again use \rf{eqn:WIsotypicalB}.  After some calculations we arrive at
\begin{align*}
\norm|\widetilde{w_{j,4}}^k|^2 &= {n-2j-2\choose k-j-1}  \|\widetilde{w_{j,4}}^{j+1}\|^2 = 2^j (n-j+1)(n-2j+1)(n-2j)^2(n-2j-1)  {n-2j-2\choose k-j-1} \\
&= 2^j (n-j+1)(n-2j+1)(n-2j)(n-j-k)(k-j) {n-2j\choose k-j}.
\end{align*}
Therefore,
\[
\fita_{j,4}^k = 
\frac{\ip<u_j^k, \widetilde{w_{j,4}}^k >}{\|u_j^k\|\cdot \|\widetilde{w_{j,4}}^k\|} 
= \sqrt{\frac {(k-j)(n-j-k)(n-j+1)} {k(n-2j)(n-2j+1)} }.
\]

One can verify that 
\begin{equation*}
(\fita_{j,1}^k)^2+(\fita_{j,2}^k)^2+(\fita_{j,3}^k)^2+(\fita_{j,4}^k)^2 = 1,
\end{equation*}
which is the content of \rf{clm:phi_j is unit}.
This identity is obvious now as they are coordinates of a unit vector in an orthonormal basis.


\section*{Acknowledgements}
We thank anonymous referees for 
their useful suggestions on the previous versions of this paper.

A.B. is supported by the Latvian Quantum Initiative under European Union Recovery and Resilience Facility project no. 2.3.1.1.i.0/1/22/I/CFLA/001.
Part of this research was supported by the ERDF project number 1.1.1.2/I/16/113.

A.R. was supported by JSPS KAKENHI Grant No.~JP20H05966 and MEXT Quantum Leap Flagship Program (MEXT Q-LEAP) Grant No.~JPMXS0120319794. Part of this work was done while he was at Kyoto University and Nagoya University as a JSPS International Research Fellow supported by the JSPS KAKENHI Grant No.~JP19F19079, and while he was at the Centre for Quantum Technologies at the National University of Singapore supported by the Singapore Ministry of Education and the National Research Foundation under grant R-710-000-012-135.

\phantomsection
\addcontentsline{toc}{section}{Bibliography}

{
\small
\bibliographystyle{habbrvM}
\bibliography{belov}
}

\appendix

\section{Upper bounds}
\label{app:upper}

In this section, we briefly describe the algorithms matching our lower bounds in \rf{thm:main}.
None of them are particularly complicated, but some can be found interesting in their own merit.
Some of the algorithms are folklore, and some are taken from~\cite{aaronson:counting}.
Unless specified otherwise, all algorithms in this section are supposed to give correct answer with probability at least $2/3$.
By \rf{rem:reduce}, it is possible to reduce error by repetition.

\subsection{Algorithmic Preliminaries}

Before we proceed, we need some well-known results both from the theory of randomized and quantum algorithms.
They all are only needed for \rf{sec:algorithms}, and are not used anywhere else in the paper.
We assume the reader is familiar with basic probability theory.

\begin{thm}[Chebyshev's Inequality]
Let $X$ is a probability distribution with variance $\sigma^2$.
It is possible to get an estimate of the mean $\bE[X]$ with additive precision $O(\sigma)$ using 1 sample from $X$.
\end{thm}

Recall that in the above theorem, as elsewhere, we assume here that the algorithm can err with probability up to $1/3$.
A distribution that is concentrated on values 0 and 1 is called Bernoulli distribution.
For that, we have an important corollary:

\begin{thm}
\label{thm:bernoulli}
Let $X$ be a Bernoulli distribution with unknown mean $p$.
Then, $O(1/(p\eps^2))$ samples from $X$ are enough to estimate $p$ with multiplicative precision $\eps$.
\end{thm}

\begin{thm}[Coupon collector]
\label{thm:coupon}
Assume we are sampling uniformly at random elements out of a set $S$ of size $s$.
The expected number of samples needed to obtain $t$ distinct elements of $S$ is
\begin{itemize}
\item $O(t)$, if $t\le s/2$;
\item $O(s\log s)$, if $t=s$.
\end{itemize}
\end{thm}

For the quantum part, we will need the following version of the quantum amplitude amplification and estimation, which follows the original design of~\cite{brassard:amplification}.

\begin{thm}
\label{thm:amplification}
Let $\phi\in\cH$ be some state known to the algorithm.
It has a decomposition of the form $\phi = \alpha \ket|\phi_0> + \beta \ket|\phi_1>$ for some orthogonal and normalised $\phi_0, \phi_1 \in \cH$ and non-negative reals $\alpha$ and $\beta$ with $\beta < 1/\sqrt{2}$.
The decomposition is unknown to the algorithm, but it has black-box access to a unitary $R$ satisfying $R\phi_0=\phi_0$ and $R\phi_1 = -\phi_1$.
Under these assumptions, the following two procedures are available:
\begin{itemize}
\item{(Amplitude Estimation)} It is possible to get an estimate of $\beta$ with multiplicative precision $\eps$ using $O(1/(\eps\beta))$ executions of $R$.
\item{(Amplitude Amplification)} 
With probability $\Omega(1)$, it is possible to get a state $\phi'_1$ such that $|\ip<\phi_1,\phi'_1>|=\Omega(1)$ using $O(1/\beta)$ executions of $R$.
\end{itemize}
\end{thm}

We will call $\phi_1$ the \emph{marked part} of the state $\phi$, and $\beta$ the \emph{marked amplitude}.

\begin{rem}
\label{rem:aboutGrover}
Amplitude amplification works by applying the operator $SR$ to $\phi$ some specific number of times, where $S$ is reflection about $\phi$.
This has two consequences:
\begin{itemize}\itemsep=0pt
\item The state of the algorithm is guaranteed to stay in the span of $\phi_0$ and $\phi_1$.
\item In the extreme case $\beta=0$, the algorithm stays in the state $\phi=\phi_0$.
\end{itemize}
\end{rem}

\begin{cor}
\label{cor:amplification}
We can use \rf{thm:amplification} in one of the following three ways.
In all of them, $\phi_1 = \psi_x = \frac{1}{\sqrt{|x|}} \sum_{i\in x}\ket |i>$ is the state from~\rf{eqn:psix}.
\begin{itemize}
\item[(a)]
(From the outside, Grover's search~\cite{grover:search} and quantum counting~\cite{brassard:counting}).
We have $\phi = \frac1{\sqrt n}\sum_{i\in[n]} \ket |i>$.
In this case, it is possible to get an element in $x$ after $O(\sqrt{n/|x|})$ executions of the input oracle, and estimate $|x|$ to multiplicative precision $\eps$ in $O\s[\frac1\eps \sqrt{n/|x|}]$ queries, where the input oracle can be any of the following: state-generating, reflecting, or membership.

\item[(b)]
(From the inside).
We have $\phi = \frac1{\sqrt{|A|}} \sum_{i\in A}\ket |i>$, where $A$ is some subset of $x$ with $|A|\le |x|/2$ which is known to the algorithm.
In this case, it is possible to get an element in $x\setminus A$ in $O(\sqrt{|x|/|A|})$ queries, and estimate $|x|$ with multiplicative precision $\eps$ in $O\s[\frac1\eps \sqrt{|x|/|A|}]$ queries, where the input oracle can be any of the following: state generating or reflecting.

Moreover, the first algorithm is guaranteed to output an element of $x$ (but possibly lying in $A$).

\item[(b')]
(Incorrect execution of (b)).
The settings are as in $(b)$, but $\phi = \ket |i>$ with $i\notin x$.
In this case, amplitude amplification will stay in the state $\ket |i>$.
\end{itemize}
\end{cor}

\pfstart
For (a), we have $\phi_0 = \frac1{\sqrt{n-|x|}} \sum_{i\notin x} \ket|i>$.
The membership oracle implements $R$ by definition.
$R$ can be also implemented as the negation of the reflecting oracle, or using 2 queries to the state-generating oracle.
The marked amplitude $\beta = \sqrt{|x|/n}$.

For (b), we have $\phi_0$ as a vector orthogonal to $\phi_1$ in the span of $\phi$ and $\phi_1$.
Again, $R$ can be implemented as the negation of the reflecting oracle, or using 2 queries to the state-generating oracle.
The marked amplitude is $\beta = \sqrt{|A|/|x|}$.

In either (a) or (b), to get an element of $x$, we use amplitude amplification and measure $\phi_1'$ in the computational basis.
Estimate of $\beta$ with multiplicative precision $\eps$ converts to an estimate of $|x|$ with multiplicative precision $O(\eps)$, hence, we can use amplitude estimation to obtain the latter.

To get the ``moreover'' part of (b), we use the first point of \rf{rem:aboutGrover} to observe that the state of the amplitude amplification algorithm stays in the span of $\phi_0$ and $\phi_1$, hence, is not supported on the elements outside of $x$.

Point (b') follows from the second point of \rf{rem:aboutGrover}, as in this case $\phi$ is orthogonal to $\phi_1$ and we have $\beta=0$.
\end{proof}

\subsection{Algorithms}
\label{sec:algorithms}
In this section, we prove that \rf{thm:main} is tight.
We first describe the algorithms, and then show how they correspond to the entries of \rf{tbl:main}.

In the following two algorithms we assume classical samples from $x$ for clarity.
Clearly, quantum samples $\psi_x$ as well as the state-generating oracle also work.

\begin{prp}
\label{prp:coupon}
It is possible to solve the approximate counting problem using $O(k\log k)$ classical samples from $x$.
\end{prp}

\pfstart
The algorithm is similar to the classical coupon collector problem.
Sample the elements out of $x$ sufficiently many times.
Output that $|x| = k$ if the number of distinct elements observed is at most $k$, otherwise output that $|x| = k'$.  The algorithm has 1-sided error.

For the analysis, we apply \rf{thm:coupon} with $t=s=k'$ to obtain an upper bound of $O(k'\log k') = O(k\log k)$.
\pfend

\begin{prp}[\cite{aaronson:counting}]
\label{prp:classical2}
It is possible to solve the approximate counting problem using $O\sB[\frac{\sqrt k}{\eps}]$ classical samples from $x$.
\end{prp}

\pfstart
The idea of the algorithm is to sample from $x$ and count the number of pairs of equal samples.
Assume we have $\ell$ classical samples: $s_1,s_2,\dots,s_{\ell}$.  
For $1\le i<j\le \ell$, let $Z_{ij} = 1_{s_i = s_j}$.
The expectation and the variance satisfy $\bE[Z_{ij}] = 1/|x|$ and $\Var[Z_{ij}] = O(1/k)$.
The events $Z_{ij}$ are not independent, but they are \emph{pairwise} independent, which allows us to write
\[
\bE\skC[\sum_{i,j} Z_{ij} ] = \frac{\ell(\ell-1)}{2|x|}
\qqand
\Var\skC[\sum_{i,j} Z_{ij}] = O\sC[\frac{\ell^2}{k}].
\]
By Chebyshev's inequality, we can distinguish whether $|x| = k$ or $|x| = k'$ if
\[
\frac{\ell(\ell-1)}{2k} - \frac{\ell(\ell-1)}{2k'} = \Omega\sC[\frac{\ell}{\sqrt{k}}].
\]
Since
\[
\frac{\ell(\ell-1)}{2k} - \frac{\ell(\ell-1)}{2k'} = \Omega\sC[\frac {\ell^2}k\eps],
\]
this happens when $\ell \ge C\sqrt{k}/\eps$ for a sufficiently large constant $C$.
\pfend

\begin{prp}[\cite{aaronson:counting}]
\label{prp:classical3}
It is possible to solve the approximate counting problem using $O\sB[\frac{n}{k\eps^2}]$ copies of the state $\psi_x$ from~\rf{eqn:psix}.
\end{prp}

\pfstart
Consider the following procedure.  Take a copy of the state $\psi_x$ and measure it against the uniform superposition $\frac1{\sqrt n}\sum_{i\in[n]} \ket|i>$.
The probability of measuring the uniform superposition is exactly $|x|/n$.
We have to detect whether this probability is $k/n$ or $(1+\eps)k/n$. 
By \rf{thm:bernoulli}, this takes $O\sB[\frac{n}{k\eps^2}]$ samples.
\pfend

\begin{prp}[\cite{aaronson:counting}]
\label{prp:algk^1/3}
It is possible to solve the approximate counting problem in $O\s[{k^{1/3}}/{\eps^{2/3}}]$ queries to the state-generating oracle.
\end{prp}

\begin{proof}
Let $t\le k/2$ be some parameter to be specified later. 
We first obtain $t$ different elements out of $x$ using $O(t)$ classical samples from $x$ by \rf{thm:coupon} (implemented via the state-generating oracle).
Next, we execute the estimation algorithm of \rf{cor:amplification}(b).
Altogether, it takes
\[
O\sC[ t + \frac1\eps \sqrt{\frac kt} ]
\]
queries to the state-generating oracle.
The optimal value of $t$ is ${k^{1/3}}/(2\eps^{2/3})\le k/2$.
\end{proof}

\begin{prp}
\label{prp:algk/eps}
Assume the algorithm is given an element $j\in x$.
Then, it is possible to solve the approximate counting problem using $O\s[\sqrt{k/\eps}]$ queries to the reflecting oracle.

Moreover, if it actually turns out that $j\notin x$, the algorithm will report this after $O(\sqrt{k})$ queries with probability 1.
\end{prp}

\begin{proof}
We repeatedly execute the search algorithm of \rf{cor:amplification}(b) starting with $A=\{j\}$, and extending $A$ with newly found elements until $|A|=t$, where $t<k/2$ is some parameter to be specified later.
If $j\in x$, we are guaranteed that $A\subseteq x$ throughout the whole procedure.

When $|A|=t$, we use the estimation algorithm of Point (b) of \rf{cor:amplification} with this choice of $A$.
The total number of queries to the input oracle is
\[
O\sC[\sqrt{k} + \sqrt{\frac k2} + \cdots + \sqrt{\frac k {t-1}} + \frac1\eps\sqrt{\frac kt} ]
= O\sC[\sqrt{kt} + \frac1\eps\sqrt{\frac kt} ].
\]
The optimal choice for $t$ is $1/(2\eps) \le k/2$.

If $j\notin x$, by Point (b') of \rf{cor:amplification}, we will always get the element $j$ back on the first execution of the search algorithm.
Therefore, if we keep on getting $j$, we can end and report that $j\notin x$.
\end{proof}

\begin{prp}
\label{prp:algreflecting2}
The approximate counting problem can be solved using either $O\s[\sqrt{\frac k\eps} + \sqrt{\strut\frac nk}]$ queries to the reflecting oracle, or $O\s[\sqrt{\frac k\eps}]$ queries to the reflecting oracle and $O\s[\sqrt{\strut\frac nk}]$ queries to the membership oracle.
\end{prp}

\pfstart
The idea is to use \rf{prp:algk/eps}, but we need an element of $x$ beforehand.
For that we use Grover's search of \rf{cor:amplification}(a).
Moreover, if Grover's search fails, the algorithm of \rf{prp:algk/eps} is able to detect this, in which case we run Grover's search again.  After $O(1)$ tries, we will succeed with high enough probability.

Grover's search requires $O(\sqrt{n/k})$ applications of the reflecting or membership input oracle.
Together with the estimate of \rf{prp:algk/eps}, this gives the required complexity.
\pfend

\begin{thm}
\label{thm:upperMain}
The \rf{tbl:main} is tight except for $k$ replaced by $k\log k$ in the first item of the first row.
\end{thm}

By this we mean that it is possible to solve the approximate counting problem if the resources allocated to the algorithm satisfy one of the eight rows of \rf{tbl:main} with the meaning of $\Omega$ as follows.
We consider Row 4 for concreteness, the other rows being similar.
There exists a universal constant $C$ such that for every choice of $\QG$ and $\ell$ satisfying $\QG\sqrt{\ell} \ge \sqrt{k}/\eps$, it is possible to solve the approximate counting problem using at most $C\ell$ copies of the state $\psi_x$ and at most $C\QG$ executions of the state-generating oracle.

\pfstart[Proof of \rf{thm:upperMain}]
We will go through the table and comment on which propositions the corresponding entries are based.

In Row 1, these are Propositions~\ref{prp:coupon}, \ref{prp:classical2} and~\ref{prp:classical3}, respectively, where, as mentioned above, $k$ is replaced by $k\log k$ in the first case.

In Row 2, it is the estimation algorithm of \rf{cor:amplification}(a).

In Row 3, the first item is from \rf{cor:amplification}(a) again, and the second one is~\rf{prp:algk^1/3}.

In Row 4, we can assume that $\ell = O(k)$ since otherwise we fall under the scope of Row 1.  By \rf{thm:coupon}, we can obtain $t = \Omega(\ell)$ distinct elements of $x$ using $\ell$ classical samples.  After that, we refer to the estimation algorithm from  \rf{cor:amplification}(b).

The first item of Row 5 is from \rf{cor:amplification}(a), and the second one is~\rf{prp:algreflecting2}.

Row 6 is obtained in the similar way as Row 4, where we use $\ell+\QG$ classical samples from $x$.

Row 7 is \rf{prp:algk/eps}.  Note that in this case it suffices to have only a single copy of $\psi_x$ to reduce error probability via repetition, \emph{cf.}~\rf{rem:reduce}.

Row 8 is~\rf{prp:algreflecting2}.
\pfend

\section{Proofs omitted from \rf{sec:adv}}
\label{app:proofs}

\mycutecommand{\subgamma}{\tilde\gamma_2}
Let us first define the subrelative and the relative $\gamma_2$-norms.

\begin{defn}[Subrelative $\gamma_2$-norm]
\label{defn:subrelative}
Let $A = (A_{x,y})$ and $\Delta = (\Delta_{x,y})$ be two families of matrices labelled by $x\in X$ and $y\in Y$.
All matrices $A_{x,y}$ are of the same dimension, and all matrices $\Delta_{x,y}$ are of the same dimension.

The \emph{subrelative $\gamma_2$-norm} is defined as 
\begin{equation}
\label{eqn:subrelative}
\subgamma(A | \Delta) = \subgamma(A_{x,y} \mid \Delta_{x,y})_{x\in X,\; y\in Y}
=\max_\Gamma\frac{\|\Gamma\circ A\|}{\|\Gamma\circ\Delta\|}
\end{equation}
\end{defn}

\begin{defn}[Relative $\gamma_2$-norm]
\label{defn:relative}
Under the same assumptions as in~\rf{defn:subrelative}, the \emph{relative $\gamma_2$-norm},
\[
\gamma_2(A | \Delta) = \gamma_2(A_{x,y} \mid \Delta_{x,y})_{x\in X,\; y\in Y},
\]
is defined as the optimal value of the following optimisation problem, where $\Upsilon_x$ and $\Phi_y$ are linear operators of appropriate size and $\cW$ is a vector space:
\begin{subequations}
\label{eqn:relative}
\begin{alignat}{2}
&\mbox{\rm minimise} &\quad& \max \sfigB{ \max\nolimits_{x\in X} \norm|\Upsilon_x|^2, \max\nolimits_{y\in Y} \norm|\Phi_y|^2 } \\
& \mbox{\rm subject to}&&  
A_{x,y} = \Upsilon_x^* (\Delta_{x,y}\otimes I_{\cW}) \Phi_y \quad \text{\rm for all $x\in X$ and $y\in Y$.}  \label{eqn:relativeCondition}
\end{alignat}
\end{subequations}
\end{defn}

These norms are related by the following simple inequality.

\mycutecommand{\diag}{\mathop{\mathrm{diag}}}
\begin{prp}
\label{prp:subgamma<gamma}
For all $A$ and $\Delta$, we have $\subgamma(A|\Delta) \le \gamma_2(A|\Delta)$.
\end{prp}

\pfstart
From~\rf{eqn:relativeCondition}, we get that for every $X\times Y$-matrix $\Gamma$:
\[
\Gamma\circ A = \diag(\Upsilon_x^*) ((\Gamma\circ\Delta)\otimes I_\cW) \diag(\Phi_y),
\]
where $\diag(\Phi_y)$ is a block-diagonal matrix with blocks $\Phi_y$ on the diagonal.
Hence,
\[
\norm|\Gamma\circ A| \le \max_{x\in X}\|\Upsilon_x\|\cdot \|\Gamma\circ\Delta\| \cdot\max_{y\in Y} \|V_y\|  \le \gamma_2(A|\Delta)\|\Gamma\circ\Delta\|.\qedhere
\]
\pfend

Let us note that the norms are equal in an important special case when $A_{x,y}$ are $1\times 1$-matrices~\cite{belovs:variations}.
In particular, we recover the usual $\gamma_2$-norm of an $X\times Y$ matrix $A = (a_{x,y})$:
\[
\gamma_2(A) 
= \subgamma(a_{x,y}\mid 1)_{x\in X, y\in Y} 
= \gamma_2(a_{x,y}\mid 1)_{x\in X, y\in Y}.
\]

Now let us prove \rf{lem:advFunctionEvaluation}.

\pfstart[Proof of~\rf{lem:advFunctionEvaluation}]
Let $T$ be the $X\times Y$-matrix defined by $T\elem[x,y] = \ip<\tau_x, \tau_y>$.
We claim that $\gamma_2(T)\le 2\sqrt\delta$.
From this the lemma follows, as
\[
\normA|\Gamma\circ E| 
=\normA|\Gamma\circ (\Xi - T)| 
\ge \normA|\Gamma\circ \Xi| - \normA|\Gamma\circ T|
\ge 3\sqrt{\delta} - 2\sqrt{\delta} = \sqrt\delta,
\]
using the dual formulation of the $\gamma_2$-norm.

Let us prove the claim.
Recall that we assume that the workspace $\cH = \bC^{2}\otimes \cH'$.
Let $\Pi_0$ and $\Pi_1$ be the projectors on the values $0$ and $1$ of the first qubit, respectively.
By condition on the error of the algorithm, $\|\Pi_1 \tau_x\|^2, \|\Pi_0\tau_y\|^2\le \delta$ for all $x\in X$ and $y\in Y$.  Thus,
\[
\ip<\tau_x, \tau_y> = 
\ip<\sqrt[4]\delta\Pi_0 \tau_x \oplus \frac{\Pi_1\tau_x}{\sqrt[4]\delta}  , \frac{\Pi_0\tau_y}{\sqrt[4]\delta} \oplus \sqrt[4]\delta\Pi_1 \tau_y>.
\]
The norm of the both vectors on the right-hand side of this equation is at most $\sqrt2\sqrt[4]{\delta}$, and the claim follows by the primal formulation of the $\gamma_2$-norm.
\pfend

Let us now move to the proofs of $\gamma_2$-equivalences from \rf{sec:advInputOracles}.
By \rf{prp:subgamma<gamma}, to prove that $A$ and $\Delta$ are $\gamma_2$-equivalent as in \rf{defn:equivalence}, it suffices to show that
\[
\gamma_2(A|\Delta),\gamma_2(\Delta,A) = O(1).
\]

\begin{prp}[\cite{belovs:distributions}]
For the state-generating input oracle $O_x$ as defined in~\rf{eqn:oracleStateGenerating}, we have that the family $(O_x-O_y)_{x,y\in D}$ is $\gamma_2$-equivalent to the family $\Delta^{\psi}\oplus \Delta^{\psi^*}$ defined in~\rf{eqn:Deltapsi}.
\end{prp}

\pfstart
In particular, the state $\ket |0>$ is orthogonal to all $\psi_x$.

In one direction, we have that
\[
\psi_x - \psi_y = (O_x-O_y)\ket|0>
\qqand
\psi_x^* - \psi_y^* = -(O_x\ket|0>)^*(O_x-O_y)O_y^*,
\]
implying that
\[
\gamma_2\sA[\Delta^{\psi}_{x,y}\oplus \Delta^{\psi^*}_{x,y} \mid O_x - O_y]_{x,y\in D} \le 1.
\]

Now let us prove the opposite direction.
Let $L_x = \psi_x \oplus \psi_x^* = \ket|\psi_x>\langle 0| + \ket|0>\langle \psi_x|$.
Note that $L_x - L_y = \Delta^{\psi}_{x,y}\oplus \Delta^{\psi^*}_{x,y}$ 
and
$L_x^2$ is the projector onto the span of $\ket|0>$ and $\ket|\psi_x>$.
The latter gives us that $O_x = L_x +I - L_x^2$.
Hence,
\[
O_x - O_y = L_x - L_y + L_y^2 - L_x^2 = (L_x - L_y) - (L_x - L_y)L_y - L_x(L_x - L_y),
\]
which implies
\[
\gamma_2\sA[O_x - O_y \mid \Delta^{\psi}_{x,y}\oplus \Delta^{\psi^*}_{x,y} ]_{x,y\in D} 
= \gamma_2\sA[O_x - O_y \mid L_x - L_y ]_{x,y\in D} 
\le 3.\qedhere
\]
\pfend

\begin{prp}[\cite{belovs:variations}]
For the membership oracle $O_x$ as defined in~\rf{eqn:oracleMembership}, we have that the family $(O_x - O_y)_{x,y\in D}$ is $\gamma_2$-equivalent to the family $(\Dmem_{x,y})_{x,y\in D}$ defined in~\rf{eqn:Delta_xy}.
\end{prp}

\pfstart
Let us denote
\[
O_0 = \begin{pmatrix}
1 & 0 \\ 0 &1
\end{pmatrix}
\qqand
O_1 = \begin{pmatrix}
0 & 1 \\ 1 & 0
\end{pmatrix},
\]
so that $O_x = \bigoplus_{i\in[n]} O_{x_i}$.
We have
\[
1_{x_i\ne y_i} = (O_{x_i}\ket|0>)^*(O_{x_i}-O_{y_i})\ket|0>
\]
implying that
\[
\gamma_2\sA[\Dmem_{x,y} \mid O_x - O_y]_{x,y\in D}\le 1.
\]
On the other hand,
\[
O_{x_i} - O_{y_i} = O_{x_i} 1_{x_i\ne y_i} - 1_{x_i\ne y_i}O_{y_i},
\]
which gives
\[
\gamma_2\sA[O_x - O_y \mid \Dmem_{x,y} ]_{x,y\in D}\le 2.\qedhere
\]
\pfend

\section{Proof of \rf{thm:multipleGamma2}}
\label{app:multipleProof}

The proof closely follows the proof of Theorem 10 from~\cite{belovs:variations}.

Let $\psi_{t,x}$ and $\psi^{(i)}_{t,x}$ be like in \rf{sec:multiple}.
Let $T$ be the total number of invocations of the combined input oracle.
We have $\ip<\psi_{1,x}, \psi_{1,y}> = \ip<\xi_x, \xi_y>$, and we define
$\psi_{T+1, x} = \tau_x$.  This gives
\[
\ip<\xi_x, \xi_y> - \ip<\tau_x, \tau_y>
=
\sum_{t=1}^T \sB[\ip<\psi_{t,x}, \psi_{t,y}> - \ip<\psi_{t+1,x}, \psi_{t+1,y}>]
\]

Let $\phi_{t,x}$ be the state of the algorithm just after the $t$-th application of the combined input oracle.
Let $\phi^{(i)}_{t,x}$ be the part of $\phi_{t,x}$ that was processed by $O^{(i)}_x$.
Here we consider the case when the application of the oracle is direct:
the inverse case is similar.
We have
$
\phi^{(i)}_{t,x} = (I \otimes O^{(i)}_x) \psi^{(i)}_{t,x}.
$
This gives us
\begin{align*}
\ip<\psi_{t,x}, \psi_{t,y}> -& \ip<\psi_{t+1,x}, \psi_{t+1,y}>
=
\ip<\psi_{t,x}, \psi_{t,y}> - \ip<\phi_{t,x}, \phi_{t,y}>\\
&=
\sum_{i=1}^s \sB[\ip<\psi^{(i)}_{t,x}, \psi^{(i)}_{t,y}> - \ip<\phi^{(i)}_{t,x}, \phi^{(i)}_{t,y}>]\\
&=
\sum_{i=1}^s \sB[
\ip<(I \otimes O^{(i)}_x)\psi^{(i)}_{t,x}, (I \otimes O^{(i)}_x)\psi^{(i)}_{t,y}> - 
\ip<(I \otimes O^{(i)}_x)\psi^{(i)}_{t,x}, (I \otimes O^{(i)}_y)\psi^{(i)}_{t,y}>]\\
&=
\sum_{i=1}^s 
\ip<(I \otimes O^{(i)}_x)\psi^{(i)}_{t,x}, {\sA[I \otimes(O^{(i)}_x-O^{(i)}_y)]}\psi^{(i)}_{t,y}>.
\end{align*}
Let us define
\[
u_x^{(i)} = \bigoplus_{t=1}^T (I \otimes O^{(i)}_x)\psi^{(i)}_{t,x}
\qqand
v_y^{(i)} =  \bigoplus_{t=1}^T  \psi^{(i)}_{t,y}.
\]
In particular, we have $\normA|u^{(i)}_x|^2 = \normA|v^{(i)}_x|^2 = L_x^{(i)}$.
Recall that $E\elem[x,y] = \ip<\xi_x, \xi_y> - \ip<\tau_x, \tau_y>$ and $\Delta^{(i)}_{x,y} = O_x - O_y$.
Therefore, for the properly-sized identity matrix $I$:
\[
E\elem[x,y]
=
\sum_{i=1}^s {u_x^{(i)}}^* \sA[I \otimes \Delta^{(i)}_{x,y} ] v_y^{(i)}.
\]
Now we proceed, similarly as in the proof of \rf{prp:subgamma<gamma}.
For every $D\times D$-matrix $\Gamma$, we have
\[
\Gamma\circ E = \sum_{i=1}^s \diag\sA[{u_x^{(i)}}^*] \sA[\Gamma\circ (I \otimes \Delta^{(i)})] \diag\sA[v_x^{(i)}].
\]
Therefore,
\[
\|\Gamma\circ E\| 
\le \sum_{i=1}^s \max_{x\in D} \normA|u^{(i)}_x|\cdot \normB|\Gamma\circ\Delta^{(i)}|\cdot \max_{y\in D} \normA|v^{(i)}_y|
\le \sum_{i=1}^s \normA|\Gamma\circ\Delta^{(i)}| \max_{x\in D}L^{(i)}_x. 
\]

\section{Limitations of the positive-weighted adversary}
\label{app:posWeighted}
The approximate counting problem with the usual membership oracle is a prime example of the power of the original formulation of the adversary method due to Ambainis~\cite{ambainis:adv}.
For completeness, we briefly restate the formulation of the bound and its application to approximate counting.

\begin{thm}
\label{thm:advBasic}
Let $f: \{0,1\}^n\supseteq D\to\{0,1\}$ be a (possibly partial) function.  Suppose $X\subseteq f^{-1}(1)$, $Y\subseteq f^{-1}(0)$, $m,m',\ell,\ell'>0$, and a relation $\sim$ between $X$ and $Y$ are such that
\itemstart
\item for each $x\in X$, there are at least $m$ different $y\in Y$ such that $x\sim y$;
\item for each $y\in Y$, there are at least $m'$ different $x\in X$ such that $x\sim y$;
\item for each $x\in X$ and $j\in [n]$, there are at most $\ell$ different $y\in Y$ such that $x\sim y$ and $x_j\ne y_j$;
\item for each $y\in Y$ and $j\in [n]$, there are at most $\ell'$ different $x\in X$ such that $x\sim y$ and $x_j\ne y_j$.
\itemend
Then, any quantum algorithm evaluating $f$ has to use $\Omega\s[ \sqrt{\frac{mm'}{\ell\ell'}} ]$ queries to the standard input oracle.
\end{thm}

\begin{thm}
Consider the same counting problem as in~\rf{thm:main}, but assume that the algorithm only has access to the standard membership oracle.
Then, in order to solve the problem, the algorithm has to use 
$\Omega \s[\frac1\eps\sqrt{\frac nk}]$ queries.
\end{thm}

\pfstart
Let $X$ contain all the sets of size $k$, and $Y$ contain all the sets of size $(1+\eps)k$.
Define the relation $\sim$ so that two inputs $x\in X$ and $y\in Y$ satisfy $x\sim y$ iff $x\subseteq y$.
Then, it is easy to check that
\[
m = \binom{n-k}{\eps k},\quad
m' = \binom{(1+\eps)k}{\eps k},\quad
\ell = \binom{n-k-1}{\eps k-1},\quad
\text{and}\quad
\ell' = \binom{(1+\eps)k-1}{\eps k-1}.
\]
Hence,
\[
\sqrt{\frac{mm'}{\ell\ell'}} = \sqrt{\frac{(n-k)}{\eps k} \frac{(1+\eps)k}{\eps k}} = 
\Omega \s[\frac1\eps\sqrt{\frac nk}].\qedhere
\]
\pfend

This is a particularly short and nice combinatorial proof, and the question arises whether it is possible to adapt a variant of this technique for a more general \rf{thm:main}.
We argue that it is most likely impossible, in the sense that the most natural yet quite general adaptation provably fails.

Previous research suggests that it is possible to substitute a general lower bound like in~\rf{prp:formulation} with a simple combinatorial argument like in~\rf{thm:advBasic}  when~\rf{prp:formulation} admits a solution with $\Gamma$ having only non-negative real entries, see e.g.~\cite[Section 3.2.3]{belovs:phd}.
The latter is known as \emph{positive-weighted adversary}, and it is subject to some limitations, see~\cite[Section 3.3.2]{belovs:phd}.

We will prove that it is impossible to obtain a good lower bound via~\rf{prp:formulation} if we restrict $\Gamma$ to have non-negative entries.
For simplicity, we consider the case $\eps=1$ and $n\gg k$.
Moreover, we will restrict the algorithm.
We assume the algorithm only has access to the copies of the state $\psi_x$ and the usual membership oracle, and we only consider the exact version of the problem (with no error allowed).

Repeating the reasoning of \rf{sec:technicalFormulation} but for the case when we only have the membership oracle and assuming that there is no error, we get that the lower bound is given by
\begin{alignat*}{2}
&\mbox{\rm maximise} &\quad& \| \Gamma\circ \Psi^{\circ\ell}\| \\
& \mbox{\rm subject to}&&  \|\Gamma\circ \Delta_i\|\le 1\quad\text{for all $i\in [n]$}.
\end{alignat*}
The usual dual formulation of this problem is
\begin{alignat*}{3}
&\mbox{\rm minimise} &\quad& \max_{z\in D}\sum\nolimits_{j \in [n]} X_j\elem[z,z] \\
& \mbox{\rm subject to}&& \sum\nolimits_{j\colon x_j \ne y_j} X_j\elem[x,y] = \Psi\elem[x,y]^\ell &\quad& \text{\rm for all $x\in X$ and $y\in Y$;} \\
&&& X_j\succeq 0 && \mbox{\rm for all $j\in [n]$,}
\end{alignat*}
where $D = X\cup Y$ and the optimisation is over $D\times D$ positive semi-definite matrices $X_j$.
However, if we restrict $\Gamma$ to have non-negative entries, the bound becomes (the proof is essentially the same as in~\cite[Section 3.3.2]{belovs:phd}):
\begin{subequations}
\label{eqn:pos}
\begin{alignat}{3}
&\mbox{\rm minimise} &\quad& \max_{z\in D}\sum\nolimits_{j \in [n]} X_j\elem[z,z] \\
& \mbox{\rm subject to}&& \sum\nolimits_{j\colon x_j \ne y_j} X_j\elem[x,y] \ge \Psi\elem[x,y]^\ell &\quad& \text{\rm for all $x\in X$ and $y\in Y$;} \label{eqn:posCondition}\\
&&& X_j\succeq 0 && \mbox{\rm for all $j\in [n]$.}
\end{alignat}
\end{subequations}

Now we will show that this optimisation problem admits a feasible solution with objective value $O\sB[\sqrt{\frac n{k2^\ell}}]$.
This therefore rules out a non-trivial lower bound on the number of queries to the membership oracle whenever $\ell = \omega(\log n)$.

Indeed, when $\ell=0$, the approximate counting problem can be solved in $O(\sqrt{n/k})$ queries.
This means that~\rf{eqn:pos} has a feasible solution with objective value $O(\sqrt{n/k})$ when the right-hand side of~\rf{eqn:posCondition} is replaced by 1.
(It is also not hard to come up with an explicit solution.)
Now note that 
\[
\Psi[x,y] = \ip<\psi_x, \psi_y> \le 1/\sqrt{2}.
\]
This means that, if we scale down the solution by $\sqrt{2^\ell}$, we get a feasible solution to the original variant of~\rf{eqn:pos} with the objective value $O\sB[\sqrt{\frac n{k2^\ell}}]$.

\end{document}